\input ./style/arxiv-general.cfg
\documentclass[MSNbibl,number,citesort,seceqn,dvips]{arxbj}
\makeatletter
   \@ifpackageloaded{graphicx}{}{\usepackage{graphicx}}
\makeatother
\usepackage{mathbh}
\usepackage{upgreek}
\usepackage{subeqn}
\usepackage{algorithm}
\usepackage{algorithmic}

\aid{0}
\volume{21}
\issue{3}
\pubyear{2015}
\firstpage{1304}
\lastpage{1340}
\doi{10.3150/13-BEJ578} 

\makeatletter
\newcommand{\rrVert}{\Vert}
\newcommand{\rrvert}{\vert}
\newcommand{\llVert}{\Vert}
\newcommand{\llvert}{\vert}

\newcommand{\Rset}{\mathbb{R}}
\def\param{\theta}
\def\Xset{\mathbb{X}}
\def\Tset{\Rset^d \times\mathcal{C}_d^+}
\def\TsetR{\Theta}
\def\P{\mathcal{P}}
\def\TsetC{\cK}
\newcommand{\TV}{\operatorname{TV}}
\newcommand{\tw}{w}
\renewcommand{\th}{h}
\newcommand{\mup}{\mu_{\pi_\theta}}
\newcommand{\Sigmap}{\Sigma_{\pi_\theta}}
\newcommand{\Leb}{\operatorname{Leb}}
\newcommand{\Supp}{\operatorname{Supp}}
\newcommand{\Vt}{V_\theta}
\newcommand{\Xfield}{\mathcal{X}}
\newcommand{\Sset}{\mathbb{S}}
\newcommand{\cB}{{\mathcal{B}}}
\newcommand{\cZ}{{\mathcal{Z}}}
\def\PP{\mathbb{P}}
\newcommand{\Nset}{{\mathbb{N}}}

\newtheorem{theo}{Theorem}[section]
\newtheorem{prop}[theo]{Proposition}
\newtheorem{lemm}[theo]{Lemma}
\newproclaim{assu}{Assumption}

\newremark{rem}[theo]{Remark}
\newremark{remark}{Remark}
\newcommand{\argmin}{\operatorname{arg\,min}}
\newcommand{\diag}{\operatorname{diag}}
\newcommand{\cA}{{\mathcal A}}
\newcommand{\cC}{{\mathcal C}}
\newcommand{\cF}{{\mathcal F}}
\newcommand{\cK}{{\mathcal K}}
\newcommand{\cL}{{\mathcal L}}
\newcommand{\cN}{{\mathcal N}}
\newcommand{\cP}{{\mathcal P}}
\newcommand{\cR}{{\mathcal R}}
\newcommand{\cS}{{\mathcal S}}
\newcommand{\cV}{{\mathcal V}}
\newcommand{\cU}{{\mathcal U}}
\newcommand{\cW}{{\mathcal W}}
\newcommand{\bSigma}{{\Sigma}}
\newcommand{\eps}{{\varepsilon}}
\newcommand{\transpose}{\intercal}
\newcommand{\Algo}[1]{\textsc{#1}}
\newcommand{\setto}{\leftarrow}
\newcommand{\eqref}[1]{(\ref{#1})}

\makeatother

\begin{document}
\begin{frontmatter}

\title{Adaptive MCMC with online relabeling}
\runtitle{Adaptive MCMC with online relabeling}

\begin{aug}
\author[1]{\inits{R.}\fnms{R\'emi} \snm{Bardenet}\corref{}\thanksref{1}\ead[label=e1]{remi.bardenet@gmail.com}},
\author[2]{\inits{O.}\fnms{Olivier} \snm{Capp\'e}\thanksref{2,e2}\ead[label=e2,mark]{cappe@telecom-paristech.fr}},
\author[2]{\inits{G.}\fnms{Gersende}~\snm{Fort}\thanksref{2,e3}\ead[label=e3,mark]{gfort@telecom-paristech.fr}} \and
\author[1,3]{\inits{B.}\fnms{Bal\'azs}~\snm{K\'egl}\thanksref{1,3}\ead[label=e4]{balazs.kegl@gmail.com}}
\address[1]{Laboratoire de Recherche en Informatique,
Universit\'{e} Paris-Sud XI, rue Noetzlin, 91190 Gif-sur-Yvette,
France.  \printead{e1}}
\address[2]{LTCI, Telecom ParisTech \& CNRS, 37 rue Dareau, 75013
Paris, France.\\ \printead{e2};
\printead*{e3}}
\address[3]{CNRS, Laboratoire de l'Acc\'el\'erateur Lin\'eaire,
Universit\'{e} Paris-Sud XI, 91898 Orsay, France.\\
\printead{e4}}
\end{aug}

\received{\smonth{10} \syear{2012}}
\revised{\smonth{10} \syear{2013}}

%
\begin{abstract}
When targeting a distribution that is {\it artificially} invariant
under some permutations, Markov chain Monte Carlo (MCMC) algorithms
face the \emph{label-switching} problem, rendering marginal
inference particularly cumbersome. Such a situation arises, for
example, in the Bayesian analysis of finite mixture models. Adaptive MCMC
algorithms such as adaptive Metropolis (AM), which self-calibrates its
proposal distribution using an online estimate of the covariance matrix of
the target, are no exception. To address the label-switching issue,
\emph{relabeling} algorithms associate a permutation to each MCMC sample,
trying to obtain reasonable marginals. In the case of adaptive Metropolis
(\textit{Bernoulli} \textbf{7} (2001) 223--242), an \emph{online}
relabeling strategy is required. This paper
is devoted to the AMOR algorithm, a provably consistent variant of AM that
can cope with the label-switching problem. The idea is to nest relabeling
steps within the MCMC algorithm based on the estimation of a \emph{single
covariance matrix} that is used \emph{both} for adapting the covariance of
the proposal distribution in the Metropolis algorithm step \emph{and} for
online relabeling. We compare the behavior of AMOR to similar relabeling
methods. In the case of compactly supported target distributions, we
prove a
strong law of large numbers for AMOR and its ergodicity. These are
the first results on the consistency of an online relabeling algorithm
to our
knowledge. The proof underlines latent relations between relabeling and
vector quantization.
\end{abstract}

%
\begin{keyword}
\kwd{adaptive Markov chain Monte Carlo}
\kwd{label-switching}
\kwd{stochastic
approximation}
\kwd{vector quantization}
\end{keyword}

\end{frontmatter}

\section{Introduction}
\label{s:introduction}
Markov chain Monte Carlo (MCMC) is a generic approach for exploring complex
probability distributions based on sampling \cite{RoCa04}. It has
become the {\it de facto}
standard tool in many applications of Bayesian inference. However, a
very common
situation in which MCMC algorithms face serious difficulties is when
the target
posterior distribution is known to be invariant under some permutations (or
block permutations) of the variables. In that case, the difficulties
are both
computational, as most often the MCMC algorithm fails to validly visit
all the
modes of the posterior, and inferential, in particular rendering marginal
posterior inference about the individual variables particularly cumbersome
\cite{CelHuRob00}. In the literature, this latter difficulty is usually referred
to as the \emph{label switching problem} \cite{Ste00}. The most well-known
example of this situation is when performing Bayesian inference in a mixture
model. In this case, the mixture likelihood is invariant to permuting
the mixture
components and, most often, the prior itself does not favor any specific
ordering of the mixture components
\cite{Cel98,Ste00,Jas05,JaHoSt05,PaIl10,SpJaWi10,MaMeRo04}. Another important
example arises in signal processing with additive decomposition models.
In this
case, the observed signal is represented as the superposition of exchangeable
signals, and the main goal is to recover the individual signals or their
parameters. In addition, often the number of signals also has to be
determined \cite{RooBecFl12,Roo12,BaKe12}. It was observed empirically
that when
the dimension of the model is not known, the reversible jump sampler
\cite{RiGr97} makes it easier to visit the multiple modes corresponding
to the
permutations but, of course, marginal inference becomes harder due to the
additional difficulty of associating components between models of varying
dimension.

In this contribution, we address the label switching problem in the
generic case where no useful external information on the target is
known. This
corresponds, for instance, to a posterior distribution when neither
the likelihood is assumed to have a specific form, nor the prior
is chosen to have conjugacy properties, which forbids the use of Gibbs sampling
or other specialized sampling strategies. We assume, however, that the
target is known to be invariant under some permutations of
the parameters. This framework is typical, for instance, in
experimental physics
applications where the likelihood computation is commonly deferred to a
{\it
black-box} numerical code. In those cases, one cannot assume anything about
the structure of the posterior or its conditional distributions, except that
they should be invariant to some permutations of the parameters. We also
restrict ourselves to the case where the dimension of the model is
finite and known so the parameters of the model are $\mathbb
{R}^d$-valued for some fixed and
finite $d$.

Following \cite{AtFoMoPr11}, an {\it adaptive MCMC} algorithm is an algorithm
which, given a family of MCMC transition kernels $(P_\param)_{\param\in
\TsetR}$ on a space $\mathbb{X}$, produces a ($\mathbb{X} \times
\TsetR$)-valued process $((X_n, \param_n))_{n \geq0}$ such that the
conditional distribution of the sample $X_{n+1}$ given the past is
$P_{\param_n}(X_n, \cdot)$. In practice, adaptive MCMC are MCMC
algorithms that
can self-calibrate their internal parameters along the iterations in
order to
reach decent performance without (or with almost no) knowledge about
the target
distribution, eliminating the grueling step of tuning the proposals. Adaptive
MCMC has been an active field of research in the last ten years,
following the
pioneering contribution of \cite{HaSaTa01} -- see \cite{AnTh08} as well
as the
other papers in the same special issue of \emph{Statistics and Computing},
along with \cite{AtFoMoPr11,AnRo01,RoRo09}. Adaptive Metropolis
(hereafter AM;
\cite{HaSaTa01}) and its variants aim at identifying the unknown covariance
structure of the target distribution along the run of a random walk
Metropolis--Hastings algorithm with a multivariate Gaussian proposal. The
rationale behind this approach is based on scaling results which
suggest that,
when $d$ tends to $+\infty$, the chain correlation is minimized when the
covariance matrix used in the proposal distribution matches, up to a constant
that depends on the dimension, the covariance matrix of the target, for
a large
class of unimodal target distributions with independent marginals
\cite{RoGeGi97,RoRo01}. AM thus progressively adapts, using a stochastic
approximation scheme, the covariance of the proposal distribution to the
estimated covariance of the target.

It has been empirically observed in \cite{BaCaFoKe12}, and we provide further
evidence of this fact below in Section~\ref{s:example}, that the
efficiency of
AM can be greatly impaired when label switching occurs. The reason for
such a
difficulty is obvious: if label switching occurs, the estimated covariance
matrix no longer corresponds to the local shape of the modes of the posterior
and so the exploration can be far from optimal. In Section~\ref{s:example}, we
also provide some empirical evidence that off-the-shelf solutions to the
label-switching problem, such as imposing identifiability constraints or
post-processing the simulated sample, are not fully satisfactory. A key
difficulty here is that most of the approaches proposed in the
literature are
based on post-processing of the simulated trajectories \emph{after} the MCMC
algorithm has been fully run
\cite{Ste00,Jas05,JaHoSt05,PaIl10,SpJaWi10,MaMeRo04,RooBecFl12}. Unfortunately,
in the case of adaptive MCMC, post-processing cannot solve the improper
exploration issue described above. On the other hand, online relabeling
algorithms \cite{RiGr97,CelHuRob00,CrWe11} often require manual
tuning based on, for
example, prior knowledge on the location of the redundant modes of the
target. Without such manual tuning they often yield poor samplers, as
we will
show it in Section~\ref{s:example}.

Our main purpose in this paper is to provide a provably consistent
variant of AM
that can cope with the label-switching problem. In \cite{BaCaFoKe12}, we
proposed an adaptive Metropolis algorithm with online relabeling,
called AMOR,
based on the original idea of \cite{Cel98}. The idea is to nest
relabeling steps
within the MCMC algorithm based on the estimation of a \emph{single covariance
matrix} that is used \emph{both} for adapting the covariance of the proposal
distribution used in the Metropolis algorithm step \emph{and} for online
relabeling. Contrary to \cite{Cel98}, the AMOR algorithm also corrects
for the
relabelings using a modified acceptance ratio. Similarly to
\cite{Cel98}, though, AMOR requires to loop over all possible
relabelings of
proposed points, which limits the method in practice to applications
with a
relatively small number of permutations. Modifications and heuristics
that address this issue are out of the scope of this paper.

In Section~\ref{s:example}, we provide empirical evidence that the
coupling established in AMOR between the criterion used for relabeling
and the estimation of the covariance of the local modes of the
posterior is beneficial to avoid the distortion of the
marginal distributions. Furthermore, the example considered in
Section~\ref{s:example} also demonstrates that the AMOR algorithm
samples from
nontrivial identifiable restrictions of the posterior distribution,
that is,
truncations of the posterior on regions where the posterior marginals are
distinct but from which the complete posterior can be recovered by
permutation. The study of the convergence of AMOR in
Section~\ref{s:results} reveals an interesting connection with the
problem of
optimal probabilistic quantization \cite{GrLu00}, which was implicit in earlier
works on label switching. It was observed previously by \cite{Pag97}
that some
adjustments to the usual theory of stochastic approximation are
necessary to
analyze online optimal quantification due to the presence of points
where the
mean field of the algorithm is not differentiable. To circumvent
this difficulty, we introduce the stable AMOR algorithm, a novel
variant of the AMOR algorithm that avoids these problematic points of
the parameter space. Finally, we establish consistency results for the
stable AMOR algorithm, showing that it indeed asymptotically provides samples
distributed under a suitably defined restriction of the posterior
distribution in which the parameters are marginally identifiable.

The paper is organized as follows. In Section~\ref{s:AMOR}, we
describe the stable AMOR algorithm and compare it with alternative
approaches on an
illustrative example. In Section~\ref{s:results}, we address the
convergence of the algorithm. The detailed proofs are provided in the \hyperref[a:proofs]{Appendix}.

\section{The stable AMOR algorithm}
\label{s:AMOR}
In this section, we introduce the stable AMOR algorithm and illustrate its
performance on an artificial example.

\subsection{The algorithm}

Let $\pi$ be a density with respect to (w.r.t.) the Lebesgue measure on
$\Rset^d$ which is invariant to the action of a finite group $\cP$ of
permutation matrices, that
is,
\[
\forall x\in\mathbb{R}^d, \forall P\in\cP,\quad\quad \pi(x)=\pi(Px).
\]
Denote by $\cC_d^+$ the set of $d \times d$ real positive definite matrices.
For $\theta= (\mu, \Sigma)$ with $\mu\in\Rset^d$ and $\Sigma\in\cC
_d^+$, define
$L_\param\dvtx \Rset^d\rightarrow\Rset_+$ by
%
\begin{equation}
L_\theta(x) = (x-\mu)^T\Sigma^{-1}(x-\mu),
\label{e:criterion}
\end{equation}
and let $\cN(\cdot|\mu,\Sigma)$ denote the Gaussian density with
mean $\mu$
and covariance matrix $\Sigma$.

Let $\TsetR\subseteq\Rset^d \times\cC_d^+$ and $(\TsetC_q)_{q \in
\Nset}$ be an increasing sequence of compact subsets of $\TsetR$ such that
$\bigcup_{q \in\Nset} \TsetC_q = \TsetR$.

Algorithm \ref{figPseudocodeStableAMOR} describes the pseudocode of
stable AMOR
\cite{BaCaFoKe12}. Choose $\param_0 \in\TsetC_0$.

\begin{algorithm}[top]
\caption{}\label{figPseudocodeStableAMOR}
\begin{algorithmic}
\STATE$\!\!\!\!\!\!\Algo{stableAMOR} (\pi(\cdot),
X_{0}, T, \theta_0=(\mu_{0},\bSigma_{0}), c, (\gamma_t)_{t\geq0},
\alpha, (\TsetC_{\psi})_{\psi\geq0} )$
\STATE
\STATE\ \ 1\quad\quad$\cS\setto\emptyset$
\STATE\ \ 2\quad\quad$\psi\setto0$ \quad\quad$\triangleright$ \textit
{Projection counter}
\STATE\ \ 3\quad\quad\textbf{for} $t \setto1$ \textbf{to} $T$
\STATE\ \ 4\quad\quad\quad\quad$\bSigma\setto c\bSigma_{t-1}$
$\triangleright$ \textit{scaled adaptive covariance}
\STATE\ \ 5\quad\quad\quad\quad$\tilde X \sim
\cN(\cdot| X_{t-1},\bSigma )$ \quad\quad
$\triangleright$ \textit{proposal}
\STATE\ \ 6\quad\quad\quad\quad$\displaystyle\tilde P \sim
\argmin_{P\in\cP}
L_{\theta_{t-1}} (P \tilde X )$ \quad\quad
$\triangleright$ \textit{pick an optimal permutation}
\STATE\ \ 7\quad\quad\quad\quad$\displaystyle\tilde X \setto
\tilde P \tilde X$ \quad\quad
$\triangleright$ \textit{permute}
\STATE\ \ 8\quad\quad\quad\quad\textbf{if} $\displaystyle
\frac{ \pi(\tilde X) \sum_{P }
\cN(PX_{t-1} |\tilde X,\bSigma
)} { \pi(X_{t-1}) \sum_P
\cN(P\tilde X | X_{t-1},\bSigma
)} > \cU[0,1] $ \textbf{then}
\STATE\ \ 9\quad\quad\quad\quad\quad\quad$X_{t} \setto\tilde X$ \quad
\quad$\triangleright$ \textit{accept}
\STATE10\quad\quad\quad\quad\textbf{else}
\STATE11\quad\quad\quad\quad\quad\quad$X_{t} \setto X_{t-1}$ \quad
\quad$\triangleright$ \textit{reject}
\STATE12\quad\quad\quad\quad$\cS\setto\cS\cup\{X_{t}\}$ \quad\quad
$\triangleright$ \textit{update
posterior sample}
\STATE13\quad\quad\quad\quad$ \displaystyle{\mu}_{t} \setto\mu_{t-1} +
\gamma_t (X_{t} - \mu_{t-1} ) + \alpha\gamma_t \operatorname
{Pen}_{t-1,1} $
\STATE14\quad\quad\quad\quad$\displaystyle
{\bSigma}_{t} \setto\bSigma_{t-1} + \gamma_t ((X_{t} - \mu
_{t-1})(X_{t} -
\mu_{t-1})^\transpose- \bSigma_{t-1} ) + \alpha\gamma_t \operatorname
{Pen}_{t-1,2}$
\STATE15\quad\quad\quad\quad\textbf{if} $({\mu_t},{\Sigma}_t)\notin
\TsetC_{\psi}$ \textbf{then}
\STATE16\quad\quad\quad\quad\quad\quad$(\mu_t, \Sigma_t) \setto(\mu
_0,\Sigma_0)$ \quad\quad$\triangleright$ \textit{Project back to
$\TsetC_{0}$}
\STATE17\quad\quad\quad\quad\quad\quad$\psi\setto\psi+1$ \quad\quad
$\triangleright$ \textit{Increment
projection counter}
\STATE18\quad\quad\quad\quad$\theta_t\setto(\mu_t,\Sigma_t)$.
\STATE19\quad\quad\textbf{return} $\cS$
\end{algorithmic}
\end{algorithm}

To explain the proposal mechanism of stable AMOR, let $\mu_{t-1}$ and
$\Sigma_{t-1}$
denote the sample mean and the sample covariance matrix, respectively,
at the
end of iteration $t-1$, and let
$\theta_{t-1}=(\mu_{t-1},\Sigma_{t-1})$. Let us also $\cS$ denote the
MCMC sample at the end of iteration $t-1$. At iteration $t$, a point
$\tilde{X}$ is first drawn from a Gaussian\vadjust{\goodbreak} centered at the previous
state $X_{t-1}$ and with covariance
$c\Sigma_{t-1}$, where $c$ implements the optimal scaling results in
\cite{RoGeGi97,RoRo01} discussed in Section~\ref{s:introduction}
(steps~4 and 5). Then in
steps~6 and 7, $\tilde{X}$ is replaced by
$\tilde{P}\tilde{X}$, where $\tilde{P}$ is a uniform draw over the
permutations in $\argmin_{P}L_{\theta_{t-1}}(P\tilde{X})$ that
minimize the relabeling criterion \eqref{e:criterion}.\hskip.2pt\footnote{Step 6
usually boils down to
selecting the permutation $\tilde{P}$ that minimizes $L_{\theta
_{t-1}}$. In
case of ties, however, $\tilde{P}$ should be drawn uniformly over the
set on
which the minimum is achieved.} This relabeling step makes the augmented
sample $S\cup\{\tilde{P}\tilde{X}\}$ look as Gaussian as possible
among all augmented sets $S\cup\{P\tilde{X}\}$, $P\in\cP$. Formally, it
can be seen as a projection onto the Voronoi cell $V_{\theta_{t-1}}$, where
%
\begin{equation}
\Vt= \bigl\{x\in\Xset/ L_\theta(x)\leq L_\theta(Px), \forall P
\in\cP \bigr\}.\label{def:Vtheta}
\end{equation}
Then, in steps 8 to 11, the candidate
$\tilde{P}\tilde{X}$ is accepted or rejected according to the usual
Metropolis--Hastings rule. The sample mean and covariance are adapted according
to a Stochastic Approximation (SA) scheme in steps 13 and
14; $\alpha\in[0, \infty)$ and $\operatorname{Pen}_{t,i}$ is a
penalty term used to drive the parameters $\param_t = (\mu_t,
\Sigma_t)$ toward the
set of interest $\TsetR$. In Section~\ref{s:results}, we will give
examples of
parameter set $\TsetR$ and penalty terms $\operatorname{Pen}_{t,i}$.
$(\gamma_t)_{t
\geq1}$ is a sequence of nonnegative steps, usually set according to a
polynomial decay $\gamma_t\sim\gamma_\star t^{-\beta}$ for some $\beta
\in
(1/2,1]$. Finally, steps 15 to 17
are a truncation mechanism with random varying bounds to make the SA algorithm
stable. In SA procedures, such a step is a way to make the paths
$(\theta_t)_{t \geq0}$ bounded with probability one, which is a required
property to prove the convergence of these procedures (see,
e.g., \cite{chen:2002}). We will provide in Section~\ref{s:results} sufficient
conditions implying that the number of random truncations is finite along
almost all paths $(\theta_t)_{t \geq0}$, thus implying that after a finite
number of iterations, everything happens as if steps 15
to 17 were omitted. In practice, it is often reported in
the literature that SA is stable even when these stabilization steps are
omitted.

Stable AMOR is a doubly adaptive MCMC algorithm since it is adaptive
both in
its \emph{proposal} and \emph{relabeling} mechanisms. This means that, besides
the proposal distribution, its target also changes with the number of
iterations. In Section~\ref{s:results}, we will prove that, at each iteration
$t$, AMOR implements a random walk Metropolis--Hastings kernel with stationary
distribution $\pi_\theta \propto\pi \mathbh{1}_{\Vt}$.


\begin{figure}

\includegraphics{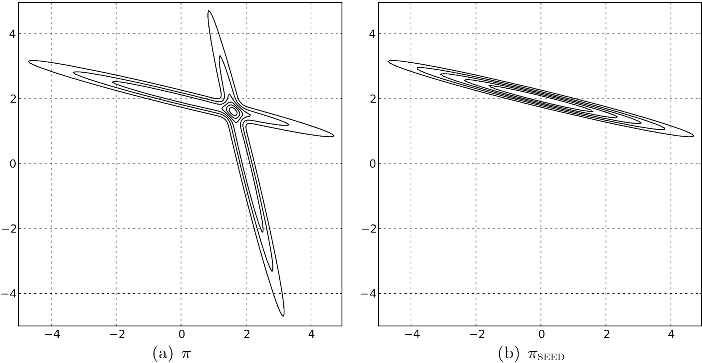}

\caption{Panel (a) shows the target distribution $\pi$ used
in Section \protect\ref{s:example},
obtained by symmetrizing the Gaussian $\pi_{\mbox{\tiny{SEED}}}$ shown in
panel (b). $\pi_{\mbox{\tiny{SEED}}}$ has mean
$(0,2)$ and covariance matrix with diagonal $(16,1)$ and
nondiagonal terms equal to $-0.975$.}
\label{f:pi}
\end{figure}

\subsection{An illustrative example}
\label{s:example}
In this section, we consider an artificial target aimed at illustrating
the gap
in performance between the stable AMOR algorithm and other common
approaches to
the label switching problem, which are compatible with adaptive MCMC. Consider
the two-dimensional p.d.f. $\pi$ depicted in Figure~\ref{f:pi}(a), which
satisfies $\pi(x)=\pi(Px)$ for $P\in\cP$, where
\[
\cP= \left\{ \pmatrix{ 1 & 0
\cr
0 & 1 }, \pmatrix{ 0 & 1
\cr
1 & 0 } \right\}.
\]
%
The density $\pi$ is a mixture of two densities with equal weights
obtained by
superposing the Gaussian p.d.f. $\pi_{\mbox{\tiny{SEED}}}$ represented in
Figure~\ref{f:pi}(b) with a symmetrized version of itself. This
artificial target does not correspond to the posterior distribution in
an actual
inference problem. In particular, although $\pi$ itself is a mixture, it
is not the posterior distribution of the parameters of any specific mixture
model. Nevertheless,\vadjust{\goodbreak} it is relevant because it is permutation invariant and
the desired solution of the label switching problem is well defined: we know
that, under suitable relabeling, we can obtain univariate near-Gaussian
marginals for both coordinates by recovering the marginals of the
two-dimensional Gaussian $\pi_{\mbox{\tiny{SEED}}}$ in
Figure~\ref{f:pi}(b). In spite of its simplicity, this example is
challenging because the two marginals of $\pi_{\mbox{\tiny{SEED}}}$
have similar means
(0 and 2) and one has large variance, which makes them hard to
separate. Given the modest dimension of the problem, we fix the number
of MCMC iterations to $20\,000$, of which $4000$ are discarded as
burn-in. For each algorithm, we assess the quality of the
relabeling strategy by looking at the corresponding restriction $\pi'$
of the target $\pi$, and we assess the efficiency of the sampling by
plotting the autocorrelation function of each sample and comparing the
sample histograms with the marginals of $\pi'$.



The results obtained when applying AM, without any relabeling, are
shown in
Figure~\ref{f:amResults}. The marginal posteriors are sampled quite well
(Figures~\ref{f:amResults}(c) and \ref{f:amResults}(d)) and the
covariance of the joint sample (indicated by a thick ellipse
Figure~\ref{f:amResults}(a)) is almost symmetric. This is not
surprising: the
joint distribution, although severely non-Gaussian, is unimodal, and
the number of iterations is large enough for AM to explore both
the original seed $\pi_{\mbox{\tiny{SEED}}}$ and its
symmetric version by frequent label switching. On the other hand, the covariance
of the joint distribution $\pi$ (Figure~\ref{f:pi}(a)) is broader than the
covariance of the seed $\pi_{\mbox{\tiny{SEED}}}$
(Figure~\ref{f:pi}(b)). This results in poor adaptive proposals
and slow mixing as indicated by the slight differences between the
marginals and
the sample marginals, and by the autocorrelation function of the first
component of the sample in Figure~\ref{f:amResults}(b). The reference
(dashed line) is the autocorrelation function of an MCMC chain with
optimal covariance (proportional to the covariance of the target)
targeting the single Gaussian
$\pi_{\mbox{\tiny{SEED}}}$ (Figure~\ref{f:pi}(b)).

%
\begin{figure*}

\includegraphics{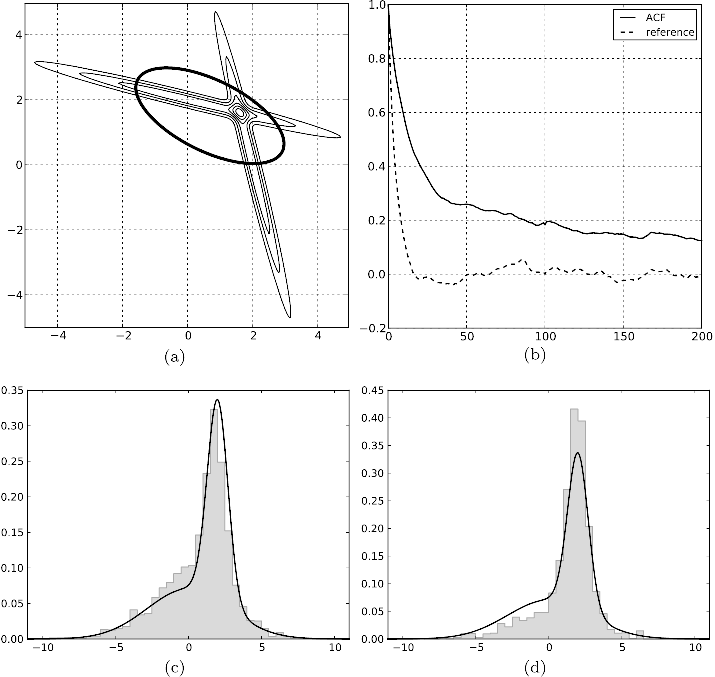}

\caption{Results of vanilla AM on the two-dimensional target $\pi$ of
Figure \protect\ref{f:pi}. The rest of the caption is the same for
Figures \protect\ref{f:amicResults} to \protect\ref{f:amorResults}. On
panel  (a), level lines of $\pi$ are depicted
in thin black
lines; a thick ellipse centered at the empirical mean $\mu_T$ of the sample
$\cS$ indicates the set $\{x \dvtx  (x-\mu_T)^T\Sigma_T^{-1}(x-\mu_T)=1\}$, where
$\Sigma_T$ is the sample covariance. When appropriate, the region of
the space
selected by (the last iteration of) the algorithm corresponds to the unshaded
background while the region not selected is shaded. On
panel~(b), the autocorrelation function (ACF)
of the first
component of $\cS$ is plotted as a solid line. The dashed line
indicates the
ACF obtained when sampling from the seed Gaussian $\pi_{\mbox{\tiny
{SEED}}}$ of
Figure \protect\ref{f:pi}(b) using a random walk Metropolis algorithm
with an optimally tuned covariance matrix. Panels  (c) and
 (d) display the histograms of the two marginal
samples. The solid curves are the marginals of $\pi$ in this figure. In
Figures \protect\ref{f:amicResults} to \protect\ref{f:amorResults},
they are the marginals of
$\pi$ restricted to the unshaded region selected by the algorithms.}
\label{f:amResults}
\end{figure*}

We now consider a modified version of AM with online relabeling
obtained by
simply ordering the variables, meaning that after each proposal $x=(x_1,x_2)$,
the components of the proposed point are permuted so that $x_1\leq
x_2$. This
strategy is known as \emph{imposing an identifiability constraint}. It
is known
to perform badly when the constraint does not respect the topology of
the target
\cite{MaMeRo04}. The results of this approach on our illustrative
example are
shown in Figure~\ref{f:amicResults}. The unshaded triangle in
Figure~\ref{f:amicResults} shows that this time the sample is
restricted to a
sub-region of $\mathbb{R}^2$ where the components are identifiable.
Unfortunately, the marginals of $\pi$ restricted to the unshaded
triangle in
Figures~\ref{f:amicResults}(c) and \ref{f:amicResults}(d) are even more
highly skewed than the marginals of the full joint distribution
$\pi$. In addition, sampling from the restricted distribution
$\pi^\prime$ is not easier than before indicated by the
autocorrelation function in Figure~\ref{f:amicResults}(b).

%
\begin{figure*}

\includegraphics{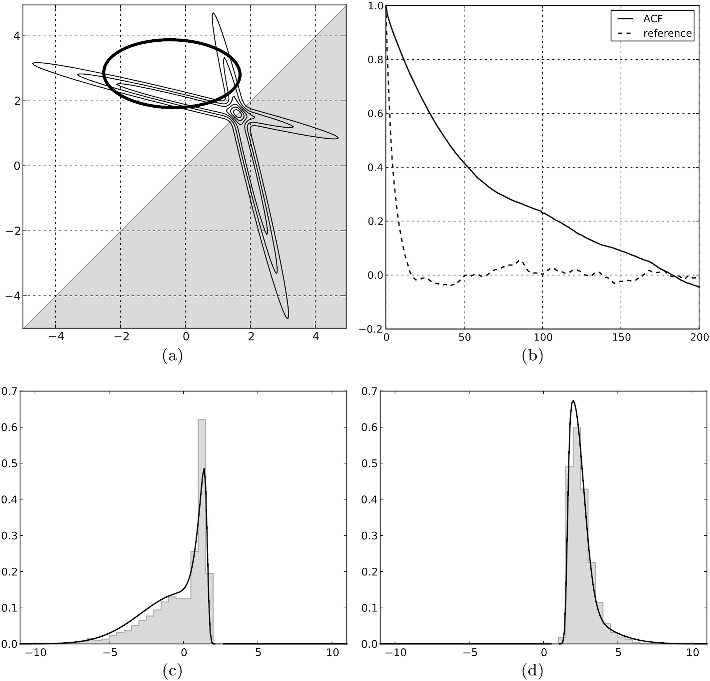}

\caption{Results of AM with online ordering constraint. For details
about the
plots, see the caption of Figure \protect\ref{f:amResults}.}
\label{f:amicResults}
\end{figure*}

\begin{figure}

\includegraphics{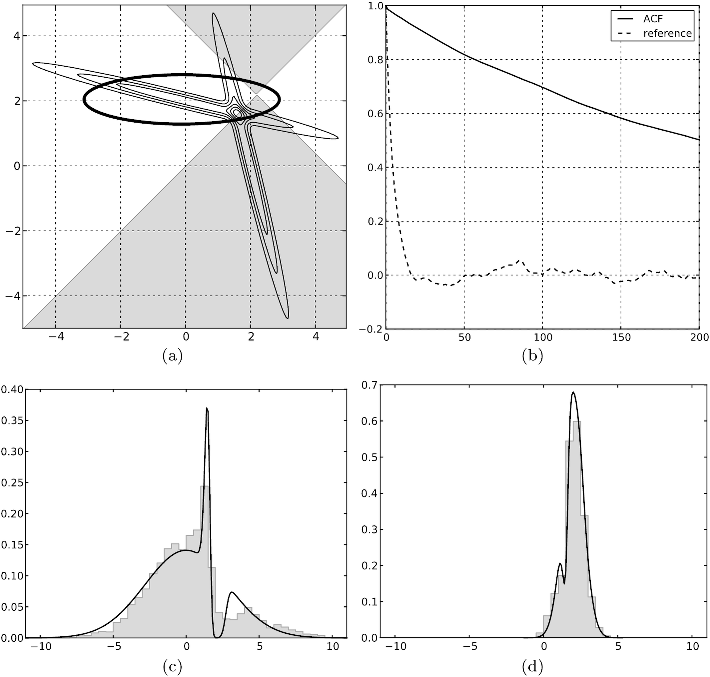}

\caption{Results of Celeux's algorithm. For details about the plots,
see the
caption of Figure~\protect\ref{f:amResults}.}
\label{f:celeuxResults}
\end{figure}

Next, we consider the approach introduced by Celeux in \cite{Cel98}. Celeux's
algorithm builds on a nonadaptive random-walk Metropolis, where
online relabeling is performed in the following way: when a point
$x=(x^{(1)}, x^{(2)})$ is proposed at time $t$, it is relabeled by
%
\begin{eqnarray}\label{e:celeux}
x &\leftarrow&\arg\min \left\{ \pmatrix{ x^{(1)}-\mu_{t}^{(1)}
\cr
x^{(2)}-\mu_{t}^{(2)} } ^T
D_t^{-1} \pmatrix{ x^{(1)}-\mu_{t}^{(1)}
\cr
x^{(2)}-\mu_{t}^{(2)} } ,\nonumber\right.
\\[-8pt]\\[-8pt]
&&\hphantom{\arg\min \left\{\right.}\left. \pmatrix{
x^{(2)}-\mu_{t}^{(1)}
\cr
x^{(1)}-
\mu_{t}^{(2)} }^T D_t^{-1}
\pmatrix{ x^{(2)}-\mu_{t}^{(1)}
\cr
x^{(1)}-
\mu_{t}^{(2)} } \right\}, \nonumber
\end{eqnarray}
where $\mu_t=(\mu_{t}^{(1)}, \mu_{t}^{(2)})$ is the
empirical mean of the current sample $x_{1:t}= x_1,\dots,x_t$ and $D_t$ is
the diagonal matrix containing the empirical variances of the
coordinates of
$x_{1:t}$ on its diagonal. Formally, this relabeling rule is equivalent
to steps
6 and 7 of Algorithm \ref{figPseudocodeStableAMOR},
but with all nondiagonal elements of $\Sigma$ equal to zero. The
results of
Celeux's algorithm are shown in Figure~\ref{f:celeuxResults}. It is
hard to
determine precisely the formal target of the algorithm. In particular,
given the non-isotropic shape of the target, we used a non-isotropic
Gaussian proposal with diagonal covariance matrix, and while the
preservation of the detailed balance
condition then requires incorporating a term into the acceptance ratio
to account for the relabeling, it is absent in this approach. It is still
possible that the algorithm is \emph{approximately} sampling
from the restriction $\pi^\prime$ of $\pi$ to this unshaded area in
Figure~\ref{f:celeuxResults} (which represents the relabeling rule implemented
at the end of the run) in a certain sense. The histograms in
Figures~\ref{f:celeuxResults}(c) and \ref{f:celeuxResults}(d)
\emph{are} in agreement with the solid line marginals. Certainly,
there are no formal guarantees that this should happen. On the other
hand, in Section~\ref{s:results} we can prove the corresponding claim
for the stable AMOR algorithm.

%

This relabeling strategy seems to recover $\pi_{\mbox{\tiny{SEED}}}$
better than the
mere ordering of coordinates as suggested by the marginal plots in
Figures~\ref{f:celeuxResults}(c) and \ref{f:celeuxResults}(d) which are
less skewed and now roughly centered at the correct values (0 and 2,
respectively). However, using a diagonal covariance $D_t$ also
generates some
distortion which results in a severely non-Gaussian, bimodal marginal
in Figure~\ref{f:celeuxResults}(c). Because of these imperfections and due to the
uncorrelated proposal, the autocorrelation in Figure~\ref{f:celeuxResults}(b)
indicates, again, a much less efficient sampling than in the case of an optimal
Metropolis chain targeting $\pi_{\mbox{\tiny{SEED}}}$.

%

The significance of Celeux's algorithm is that its adaptive relabeling rule
\eqref{e:celeux} makes it possible to resolve the permutation
invariance problem
in a nontrivial way which appears to be more adapted to the true
geometry of
the target. It is still not perfect, and, as suggested by \cite{Ste00}, one
should replace the diagonal covariance matrix in \eqref{e:celeux} by
the full
covariance matrix of the sample. However, \cite{Ste00} explored this
idea only
as a post-processing approach. A severe difficulty in this context is the
computational cost: if $T$ denotes the number of drawn samples and $p$
is the
number of permutations to which $\pi$ is invariant, the required post-processing
is a combinatorial problem with $p^T$ possible relabelings. This
eventually led
\cite{Ste00} to consider a more tractable alternative instead. More importantly
in our context, we have seen above (e.g., in Figure~\ref{f:amResults}) that
running an adaptive MCMC on the full permutation-invariant target may
result in
a poor mixing performance. To achieve both relevant relabeling and efficient
adaptivity, the key idea of stable AMOR is to link the covariance of
the proposal
distribution and the covariance used for relabeling, which are
proportional to
each other in stable AMOR.

\begin{figure*}

\includegraphics{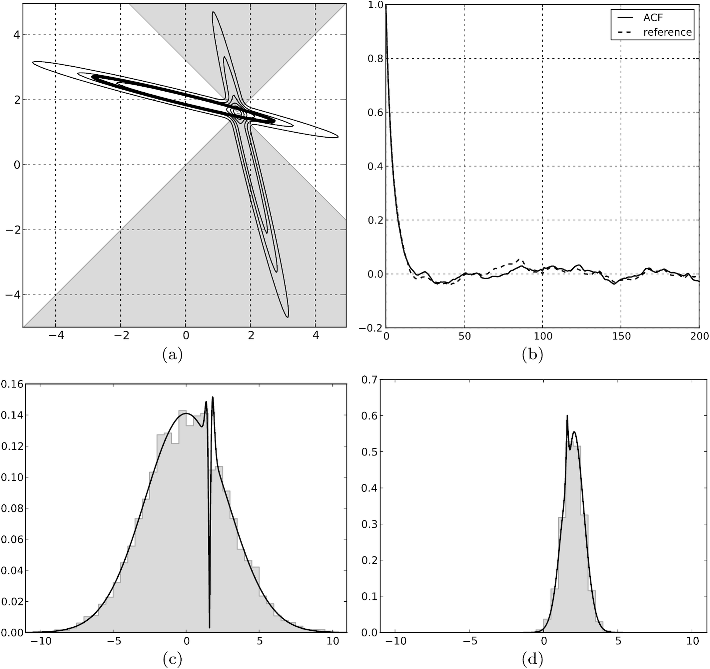}

\caption{Results of stable AMOR. For details about the plots, see the
caption of Figure
\protect\ref{f:amResults}.}
\label{f:amorResults}
\end{figure*}

Figure~\ref{f:amorResults} displays the results obtained using stable
AMOR on our
running example. Stable AMOR does separate $\mathbb{R}^2$ in two
regions that respect
the topology of the target much more closely than the approaches examined
previously. Figure~\ref{f:amorResults}(a) indicates that the relabeled target
is as Gaussian as possible among all partitionings based on a quadratic
criterion of the form \eqref{e:criterion}. The marginal histograms in
Figures~\ref{f:amorResults}(c) and \ref{f:amorResults}(d) now look almost
Gaussian. They closely match the marginals of both the restricted distribution
$\pi^\prime$ and the seed distribution $\pi_{\mbox{\tiny{SEED}}}$ in
Figure~\ref{f:pi}(b). Furthermore, the autocorrelation function of
stable AMOR (Figure~\ref{f:amorResults}(b)) is as good as the reference
autocorrelation function corresponding to an optimally tuned random walk
Metropolis--Hastings algorithm targeting the seed Gaussian $\pi_{\mbox
{\tiny{SEED}}}$ in
Figure~\ref{f:pi}(b). This perfect adaptation is possible because
the sample covariance now matches the covariance of the target
restricted to the
unshaded region of the plane (Figure~\ref{f:amorResults}(a)).


On this example, the stable AMOR algorithm thus automatically achieves,
without any
tuning, a satisfactory result that cannot be obtained with any of the methods
examined previously. Further examples of the behavior of stable AMOR are
given in the supplemental article \cite{6}. We are now ready to prove our main
result which shows that,
under suitable conditions, stable AMOR indeed asymptotically
samples from the target distribution restricted to a region on which the
marginals are identifiable, and that the sample mean and covariance
converge to
the corresponding moments of the restricted target.

\section{Convergence results}
\label{s:results}

We prove the convergence of stable AMOR under the following
condition on $\pi$.
%
\begin{assu}
\label{a:target} $\pi$ is a density
w.r.t. the Lebesgue measure on $\Rset^d$, which is bounded and with compact
support $\Xset$, and which is invariant to permutations in the group $\cP$:
\[
\forall x\in\Xset, \forall P\in\cP,\quad\quad \pi(Px)=\pi(x).
\]
\end{assu}

This section is organized as follows. We first describe which version
of the
stable AMOR algorithm we consider, and we show that it is an adaptive
MCMC algorithm. We
then characterize the limiting behavior of the sequence $(\theta_t)_{t
\geq0}$
(see Theorem~\ref{th:thetaConvergence}) and address a strong law of large
numbers for the samples $(X_t)_{t \geq0}$, as well as the ergodicity
of the
sampler (see Theorems \ref{th:xConvergence} and \ref{th:ergodicity}). All
proofs are given in the \hyperref[a:proofs]{Appendix}.

We are interested in finding a subset $\Vt$ of $\Xset$ of the form
(\ref{def:Vtheta}) such that the cells $(P \Vt)_{P \in\cP}$ cover
$\Xset$. We
will also ask that for any $P, Q \in\cP$, $P \neq Q$, the Lebesgue
measure of
$P \Vt\cap Q \Vt$ is null. Therefore, we choose the parameter set
$\TsetR$ as
follows (see Lemma~\ref{l:lemmaFromAistats} in the \hyperref[a:proofs]{Appendix}):
%
\begin{equation}
\label{eqnTheta} \TsetR= \bigl\{(\mu, \Sigma) \in\Rset^d \times
\cC_d^+ / \forall P \in\P^*, 
\Sigma^{-1}\mu \neq
P \Sigma^{-1}\mu \bigr\},
\end{equation}
where $\cP^*= \cP\setminus\{\operatorname{Id}\}$. The set $\Tset$ is
endowed with
the scalar product $\langle(a,A),(b,B)\rangle= a^Tb +
\operatorname{Trace}(A^TB)$. We will use the same notation $\|\cdot\|
$ for the norm
induced by this scalar product, for the Euclidean norm on $\Rset^d$,
and for
the norm $\| A\|=\operatorname{Trace}(A^TA)^{1/2}$ on $d\times d$ real
matrices.

Since we want to drive the parameter toward the set $\TsetR$, we
address the
convergence of the stable AMOR when $\alpha>0$ and the penalty term is
given by
%
\begin{eqnarray}
\operatorname{Pen}_{t,1} & =& - \sum_{ P \in\P^*}
\frac{1}{\|(I-P)\Sigma^{-1}_t \mu_t\|^4} U_P\Sigma^{-1}_t
\mu_t , \label{eq:penalty:mu}
\\
\operatorname{Pen}_{t,2} & =& \sum_{ P \in\P^*}
\frac{1}{\|
(I-P)\Sigma^{-1}_t\mu_t\|^4} \bigl( \mu_t \mu_t^T
\Sigma^{-1}_t U_P + U_P
\Sigma^{-1}_t\mu_t\mu^T_t
\bigr), \label{eq:penalty:var}
\end{eqnarray}
where $U_P= (I-P)^T (I-P)$. For the stabilization step, we consider the
sequence of compact sets $(\TsetC_{\delta_q})_{q \geq0}$ where
%
\begin{equation}
\label{eqnKDelta} \TsetC_\delta= \Bigl\{(\mu, \Sigma) \in\TsetR\dvtx  \inf
_{P \in\P^*} \bigl\|(I-P)\Sigma^{-1}\mu\bigr\|\geq\delta \Bigr\}
,
\end{equation}
and $(\delta_q)_{q \geq0}$ is any decreasing positive sequence such that
$\lim_{q \to\infty} \delta_q =0$ and $\TsetC_{\delta_0}$ is not empty.

Stable AMOR can be cast into the family of adaptive MCMC algorithms, in which
the updating rule of the design parameter relies on a stochastic approximation
scheme. Adaptive MCMC can be described as follows: given a family of transition
kernels $(P_\theta)_{ \theta\in\TsetR}$, the algorithm produces a
$(\Xset
\times\Theta)$-valued process $((X_t, \theta_t))_{t \geq0}$ such that the
conditional distribution of $X_t$ given its past history $X_1, \ldots, X_{t-1}$
is given by the transition kernel $P_{\theta_{t-1}}(X_{t-1},\cdot)$. This
algorithm is designed so that when $t$ tends to infinity, the
distribution of
$X_t$ converges to the invariant distribution of the kernel $P_{\theta_t}$.
Sufficient conditions for the convergence of such adaptive procedures were
recently proposed by \cite{RoRo07,FoMoPr12}. In particular, \cite{RoRo07}
provided sufficient conditions in terms of the so-called containment condition
and diminishing adaptation. Furthermore, \cite{FoMoPr12} showed that
when each
transition kernel $P_\theta$ has its own invariant distribution $\pi
_\param$,
an additional condition on the convergence of these distributions is also
required. We prove below that in our settings, each transition kernel
of stable
AMOR has its own invariant distribution; and this additional condition is
satisfied as soon as $(\param_t)_{t \geq0}$ converges almost surely.
In order
to establish this property, we will resort to convergence results for
stochastic approximation algorithms.

As a preliminary step for the convergence of stable AMOR, the stability
and the
convergence of the design parameter sequence $(\theta_t)_{t\geq0}$ is
established. Sufficient conditions for the convergence of stochastic
approximation procedures rely on the existence of a (sufficiently regular)
Lyapunov function on $\TsetR$, on the behavior of the mean field at the
boundary of the parameter set $\TsetR$, and on the magnitude of the step-size
sequence $(\gamma_t)_{ t \geq0}$.

The compactness assumption (Assumption~\ref{a:target}) makes it simpler
to analyze the
limiting behavior of the algorithm. The noncompact case is far more technical
and will not be addressed in this paper; see, e.g., \cite{FoMoPr12}
(respectively \cite{AnMoPr05}, Section~3) for examples of convergence of
adaptive MCMC (respectively a stochastic approximation procedure) when the
support of $\pi$ is not compact (respectively when the controlled
Markov chain
dynamics is not compactly supported).

Let us prove that stable AMOR is an adaptive MCMC algorithm. For any
$\theta
\in\TsetR$, define the transition kernel $P_\theta$ on $(\Xset,\Xfield
)$ by
%
\begin{equation}
P_\param(x,A) = \int_{A \cap V_\param} \alpha_\param(x,y)
q_\param(x,y) \,\mathrm{d}y + \mathbh{1}_A(x) \int
_{V_\param} \bigl(1 - \alpha_\param(x,z) \bigr)
q_\param(x,z) \,\mathrm{d}z , \label{e:kernelP}
\end{equation}
where $V_\param$ is given by (\ref{def:Vtheta}),
%
\begin{equation}
\label{eq:def:alpha} \alpha_\param(x,y) = 1 \wedge\frac{\pi(y) q_\param
(y,x) }{\pi(x) q_\param(x,y)}
\end{equation}
and
%
\begin{equation}
\label{eq:def:qtheta} q_\param(x,y) = \sum_{P \in\P}
\mathcal{N}( P y | x, c \Sigma) .
\end{equation}
For $\theta\in\Theta$, define also
%
\begin{equation}
\pi_\theta= |\cP|\mathbh{1}_{V_\theta} \pi. \label{def:pitheta}
\end{equation}
The following proposition shows that $q_\theta(x,\cdot)$ is a density on
$V_\param$ and, the distribution $\pi_\theta$ given by \eqref
{def:pitheta} is
invariant for the transition kernel $P_\theta$. It also establishes
that stable
AMOR is an adaptive MCMC algorithm: given $(X_{t-1}, \param_{t-1})$,
$X_t$ is
obtained by one iteration of a random-walk Metropolis--Hastings
algorithm with
proposal $q_{\param_{t-1}}$ and invariant distribution $\pi_{\param_{t-1}}$.
%
\begin{prop}
\label{prop:poisson} Under Assumption~\ref{a:target}, the following
assertions hold:
\begin{enumerate}[(3)]
\item[(1)] For any $\param\in\TsetR$ and $x \in\Xset$, $\int_{V_\theta}
q_\theta(x,y) \,\mathrm{d}y =1$.
\item[(2)] For any $\param\in\TsetR$, $\pi_\param P_\param= \pi_\param$
and for
any $x \in V_\param$, $P_\param(x,V_\param) =1$.
\item[(3)] Let $(\param_t, X_t)_{t \geq0}$ be given by
Algorithm \ref{figPseudocodeStableAMOR}. Conditionally on
$\sigma(X_0,\theta_0,X_1,\theta_1,\ldots,\break X_{t-1},\theta_{t-1})$, the
distribution of $X_t$ is $ P_{\theta_{t-1}}(X_{t-1},\cdot)$.
\end{enumerate}
\end{prop}
Note that the proof of Proposition~\ref{prop:poisson} is independent of the
update scheme of $(\theta_t)_{t\geq0}$, which makes the proposition valid
whatever the choice of $\alpha \operatorname{Pen}_{t,i}$.

Denote by $\cS_d$ the set of $d \times d$ symmetric real matrices. Let
$\alpha>0$ be fixed and define $H\dvtx \Xset\times\TsetR\rightarrow
\mathbb{R}^d\times\cS_d$ by
%
\begin{equation}
\label{eq:def:H} H(x,\theta) = \bigl(H_\mu(x,\theta),
H_\Sigma(x,\theta) \bigr)
\end{equation}
where
\begin{eqnarray}
H_\mu(x,\theta) &=& x-\mu- \alpha\sum_{ P \in\P^*}
\frac{1}{\|(I-P)\Sigma^{-1}\mu\|^4} U_P\Sigma^{-1} \mu,
\nonumber
\\
H_\Sigma(x,\theta) &=& (x-\mu) (x-\mu)^T - \Sigma
\nonumber
\\
&&{}+ \alpha\sum_{ P \in\P^*} \frac{1}{\|(I-P)\Sigma^{-1}\mu
\|^4} \bigl(
\mu\mu^T\Sigma^{-1}U_P + U_P
\Sigma^{-1}\mu\mu^T \bigr).
\nonumber
\end{eqnarray}

Let
%
\begin{eqnarray}
\mup& =& \int x \pi_\theta(x) \,\mathrm{d}x , \label{def:mutheta}
\\
\Sigmap&=& \int(x-\mup) (x-\mup)^T \pi_\theta(x) \,\mathrm{d}x ,
\label{def:labeltheta}
\end{eqnarray}
be the expectation and covariance matrix of $\pi_\param$, respectively. Define
the {\em mean field} $h\dvtx \TsetR\rightarrow\mathbb{R}^d\times\cS_d$ by
%
\begin{equation}
h(\theta) = \bigl(h_\mu(\theta), h_\Sigma(\theta) \bigr),
\label{def:h}
\end{equation}
where
\begin{eqnarray}
h_\mu(\theta) &=& \mu_{\pi_\theta}-\mu- \alpha\sum
_{ P \in\P^*} \frac{1}{\|(I-P)\Sigma^{-1}\mu\|^4} U_P\Sigma^{-1}
\mu,
\nonumber
\\
h_\Sigma(\theta) &=& \Sigma_{\pi_\theta} - \Sigma+ (
\mu_{\pi_\theta} - \mu) (\mu_{\pi_\theta} - \mu)^T
\nonumber
\\
&&{} + \alpha\sum_{ P \in\P^*} \frac{1}{\|(I-P)\Sigma^{-1}\mu\|
^4} \bigl( \mu
\mu^T\Sigma^{-1}U_P + U_P
\Sigma^{-1}\mu\mu^T \bigr).
\nonumber
\end{eqnarray}
The key ingredient for the proof of the convergence of the sequence
$(\param_t)_{t \geq0}$ is the existence of a Lyapunov function $w$ for
the mean
field $h$: we prove in the \hyperref[a:proofs]{Appendix} (see Lemma~\ref{l:lyapunov}) that the
function $w\dvtx \TsetR\rightarrow\mathbb{R}_+$, defined by
%
\begin{eqnarray}
\label{e:defLyapunov} w(\theta) = - \int\log\cN(x|\theta)\pi_\theta (x)
\, \mathrm{d}x
+ \frac{\alpha}{2} \sum_{P\in\cP^*} \frac{1}{\|
(I-P)\Sigma^{-1}\mu\|^2},
\end{eqnarray}
is continuously differentiable on $\TsetR$ and satisfies $\langle
\nabla w, h
\rangle\leq0$. In addition, $\langle\nabla w(\param), h(\param)
\rangle= 0$
if and only if $\param$ is in the set
%
\begin{equation}
\label{eq:def:setL} \mathcal{L} = \bigl\{\param\in\TsetR\dvtx  h(\theta)=0 \bigr\} = \bigl
\{\param\in\TsetR\dvtx  \nabla w(\theta)=0 \bigr\} .
\end{equation}
The convergence of the sequence $(\param_t)_{t \geq0}$ is proved by verifying
the sufficient conditions for the convergence of the stochastic approximation
for Lyapunov stable dynamics given in \cite{AnMoPr05}. The first step
is to
prove that the sequence is bounded with probability one: we prove that, almost
surely, the number of projections $\psi$ is finite so that the projection
mechanism (steps 15 to 17 in
Algorithm \ref{figPseudocodeStableAMOR}) never occurs after a (random) finite
number of iterations. We then prove the convergence of the stable
sequence. To
achieve that goal, following the same lines as in \cite{AnMoPr05}, we
make the
following assumption.
%
\begin{assu}
\label{a:limitSet}
Let $\cL$ be given by \eqref{eq:def:setL}. There exists $M_\star>0$
such that $\cL\subset\{\param: w(\param) \leq
M_\star\}$, and $w(\cL)$ has an empty interior.
\end{assu}
For $x\in\Rset^d$ and $A\subset\Rset^d$, define $\mathrm{d}(x,A)=
\inf_{a\in A}\| x-a\|$. The following result is proved in the \hyperref[a:proofs]{Appendix}.
%
\begin{theo}
\label{th:thetaConvergence}
Let $\beta\in(1/2,1]$ and $\gamma_\star>0$. Let $(\theta_t)_{t\geq
0}$ be the
sequence produced by Algorithm \ref{figPseudocodeStableAMOR} with
$\alpha>0$,
the penalty term given by \eqref{eq:penalty:mu} and \eqref{eq:penalty:var}, the
compact sets $\TsetC_\delta$ given by (\ref{eqnKDelta}) and $\gamma_t
\sim
\gamma_\star t^{-\beta}$ when $t\rightarrow+\infty$. Under Assumptions
\ref{a:target} and \ref{a:limitSet},
\begin{enumerate}[(2)]
\item[(1)]  The sequence $(\param_t)_{t \geq0}$ is stable:
almost surely, there exist $M>0$ and $t_\star>0$ such that for any $t
\geq
t_\star$, $\param_t \in\{\param\in\TsetR\dvtx  w(\param) \leq M \}$. In
addition, the number of projections is finite almost surely.
\item[(2)]  Almost surely, $(w(\param_t))_t$ converges to
$w^\star\in w(\cL)$ and $\limsup_t
\mathrm{d}(\theta_t,\cL_{w^\star})\rightarrow0$ where $\cL_{w^\star} =
\{\param\in\cL, w(\param) = w^\star\}$.
\end{enumerate}
\end{theo}
Theorem~\ref{th:thetaConvergence} states the convergence of $(\theta
_t)_{t\geq
0}$ to the set $\cL$ of the zeros of $h$; note that this set neither depends
on the initial values $(\param_0, X_0)$ nor on other design parameters.
In our
experiments, we always observed pointwise convergence. This is a hint
that, in practice, $\cL$ does not contain accumulation points. We now
state a strong
law of large numbers for the samples $(X_t)_{t \geq0}$.
%
\begin{theo}
\label{th:xConvergence}
Let $\beta\in(1/2,1]$, $\gamma_\star>0$, and $\param^\star\in\cL$. Let
$(X_t, \theta_t)_{t\geq0}$ be the sequence generated by
Algorithm \ref{figPseudocodeStableAMOR} with $\alpha>0$, the penalty term
given by \eqref{eq:penalty:mu} and \eqref{eq:penalty:var}, the
compact sets $\TsetC_\delta$ given by (\ref{eqnKDelta}) and $\gamma_t
\sim
\gamma_\star t^{-\beta}$ when $t\rightarrow+\infty$. Under Assumptions
\ref{a:target} and \ref{a:limitSet}, on the set $\{\lim_t \param_t =
\param^\star\}$, almost surely,
\[
\lim_{T\rightarrow\infty} \frac{1}{T} \sum
_{t=1}^T f(X_t) =\pi_{\param
^\star}(f)
,
\]
for any bounded function $f$.
\end{theo}

It is easily checked (by using Lemma~\ref{l:lemmaFromAistats}) that,
when the
function $f$ is invariant to permutations in the group $\mathcal{P}$,
$\pi_{\theta}(f) = \pi(f)$ for any $\theta\in\TsetR$. A careful
reading of
the proof of this theorem (see the remark in Section~\ref{sec:thcvg}) shows
that for such a function $f$, when the sequence $(\theta_t)_{t \geq0}$ is
stable but does not necessarily converge, it holds, almost surely,
\[
\lim_{T\rightarrow\infty} \frac{1}{T} \sum
_{t=1}^T f(X_t) =\pi(f) .
\]
Finally, Theorem~\ref{th:ergodicity} yields the ergodicity of stable AMOR.
%
\begin{theo}
\label{th:ergodicity}
Let $\beta\in(1/2,1]$, $\gamma_\star>0$, and $\param^\star\in\cL$. Let
$(X_t, \theta_t)_{t\geq0}$ be the sequence generated by
Algorithm \ref{figPseudocodeStableAMOR} with $\gamma_t \sim\gamma_\star
t^{-\beta}$ when $t\rightarrow+\infty$. Under Assumptions \ref
{a:target} and
\ref{a:limitSet},
\[
\lim_{t\rightarrow\infty} \sup_{\| f\|_{\infty}\leq1} \Bigl\llvert
\mathbb{E} \bigl[ f(X_t) \mathbh{1}_{\lim_q\theta_q=\theta^\star} \bigr] -
\pi_{\theta^\star}(f) \PP\Bigl( \lim_q\theta_q=
\theta^\star\Bigr) \Bigr\rrvert = 0.
\]
\end{theo}
Here again, a careful reading of the proof shows that when $f$ is
invariant to
permutations in the group $\mathcal{P}$, we have (see the \hyperref[r2]{Remark} in
Section~\ref{sec:thergo})
\[
\lim_{t\rightarrow\infty} \bigl\llvert \mathbb{E} \bigl[ f(X_t)
\bigr] - \pi(f) \bigr\rrvert = 0.
\]

The expression \eqref{e:defLyapunov} of $w$ provides insight into the links
between relabeling and vector quantization \cite{GrLu00}. The first
term is
similar to a distortion measure in vector quantization as noted in~\cite{BaCaFoKe12}. It can also be seen as the cross-entropy
between $\pi_\theta$ and a Gaussian with parameters $\theta$. The
second term in
\eqref{e:defLyapunov} is similar to a barrier penalty in continuous optimization
\cite{BoVa04}. From this perspective, Algorithm \ref{figPseudocodeStableAMOR}
can be seen as a constrained optimization procedure that minimizes the
cross-entropy. In that sense, if $\theta^\star$ denotes a solution to this
optimization problem, the {\it relabeled target} $\pi_{\theta^\star}
\propto
\mathbh{1}_{V_{\theta^\star}}\pi$ is the restriction of $\pi$ to one of
its symmetric
modes $V_{\theta^\star}$ that looks as Gaussian as possible among all
such restrictions.

Vector quantization algorithms have already been investigated using stochastic
approximation tools \cite{Pag97}. However, stability was guaranteed in previous
work by making strong assumptions on the trajectories of the process
$(\theta_t)_{t\geq0}$, such as in \cite{Pag97}, Theorem~32; see also
\cite{Pag97}, Results 33--37 and Remark~38. These assumptions ensure that
$(\theta_t)$ stays asymptotically away from sets where the function used
elsewhere as a Lyapunov function is not differentiable. In this paper,
we adopt
a different strategy by introducing the modifications of the stable AMOR
algorithm and adding a barrier term in the definition of our Lyapunov function
\eqref{e:defLyapunov} that penalizes these sets. One of the
contributions of
this paper is to show that this penalization strategy leads to a stable
algorithm, without requiring any strong assumption on $(\theta_t)$.

\section{Conclusion}
\label{s:conclusion}

We illustrated stable AMOR, an adaptive Metropolis algorithm with online
relabeling and proved that a strong law of large numbers holds for this
sampler. The
stable version of AMOR, given in Algorithm \ref{figPseudocodeStableAMOR},
coincides with AMOR (proposed in \cite{BaCaFoKe12}) when the penalty
coefficient $\alpha$ is set to zero and no reprojection is performed. In
practice, we observed that stable AMOR is very robust to the choice of
$\alpha$. Figure~\ref{f:stableAmorOnToyExample} illustrates this
robustness on
the toy example of Section~\ref{s:example}.

\begin{figure}[b]

\includegraphics{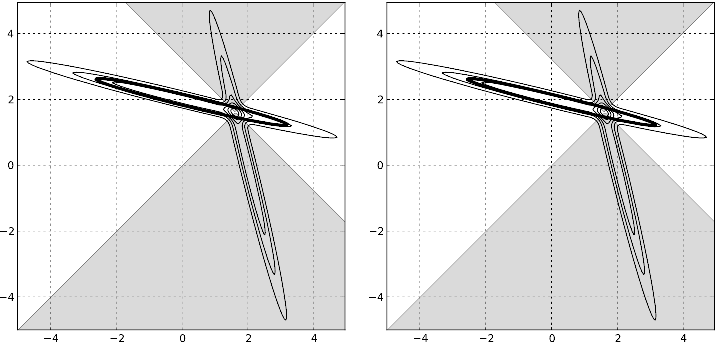}

\caption{Results of stable AMOR on the toy example of
Section \protect\ref{s:example}, with $\delta_q=10^{-2}2^{-q}$, and
$\alpha=10^{-3}$ (left) and $\alpha=1$ (right).}\label{f:stableAmorOnToyExample}
\end{figure}

Our algorithm adapts both its proposal and its target on the fly, which
makes it a
turn-key algorithm. Our results lead to a sound characterization of the target
of stable AMOR that does not depend on the initialization of the
algorithm nor on the
user. This is the first theoretical analysis of an online relabeling
algorithm to our knowledge. The proof further shows how relabeling is
related to vector quantization. Unlike previous work on stochastic
approximation schemes for vector quantization, we make no strong
assumptions on the trajectories of the process considered, rather, we
ensure that the appropriate constraint is satisfied by introducing
penalization directly into the stochastic approximation framework.

We now examine possible directions for future work. First, following our
analysis in Section~\ref{s:results}, the question of the control of the
convergence of stable AMOR arises, and proving a central limit theorem
would be
a natural next step. Second, the online nature of stable AMOR makes it cheaper
than its post-processing counterpart, but it still requires to sweep
over all
elements of $\cP$ at each iteration. This is prohibitive in problems
with large
$|\cP|$, such as additive models with a large number of
components. In
future work, we will concentrate on algorithmic modifications to reduce this
cost, potentially inspired by {\it probabilistic} relabeling algorithms
\cite{Jas05,SpJaWi10}, while conserving our theoretical results. Third,
we are
interested in extending stable AMOR to trans-dimensional problems, such as
mixtures with an unknown number of components. Reversible jump MCMC (RJMCMC;
\cite{Gre95}) also suffers from label-switching and inferential difficulties.
We will study algorithms that combine RJMCMC and stable AMOR.

\begin{appendix}
\section*{Appendix: Proofs}
\label{a:proofs}
Throughout the proof, let $\Delta_\pi>0$ be such that
%
\begin{equation}
\label{eq:def:DeltaPi} x \in\Xset\Rightarrow\|x\| \leq\Delta_\pi.
\end{equation}

For any function $f:D\rightarrow\Rset$, we will denote by $\|
f\|_\infty=\sup_{x \in D} |f(x)|$.

\subsection{Preliminary results}
We restate (with a slight adaptation) Lemma~1 of the supplementary material
from \cite{BaCaFoKe12} that we will use extensively.
%
\begin{lemm}
\label{l:lemmaFromAistats}
Let $\theta\in\TsetR$.
\begin{enumerate}[(2)]
\item[(1)] The sets $\{PV_\param,P\in\cP\}$ cover $\Xset$, and for any
$P,Q\in\cP$
such that $P\neq Q$, the Lebesgue measure of $PV_\theta\cap QV_\theta$ is
zero.
\item[(2)] Let $\lambda$ be a measure on $(\Xset,\Xfield)$ with a density
w.r.t. the
Lebesgue measure. Furthermore, let $\lambda$ be such that for any
$A\in\Xfield$ and $\P\in\cP$, $\lambda(PA)=\lambda(A)$. Then
$\lambda(V_\theta) = \lambda(\Xset)/|\cP|$.
\end{enumerate}
\end{lemm}
\begin{pf} The proof is along the lines of Lemma~1 of the
supplementary material in \cite{BaCaFoKe12}, and it is thus
omitted. It can be found in the supplemental article to the present paper~\cite{6}.
\end{pf}

\subsection{Proof of Proposition \texorpdfstring{\protect\ref{prop:poisson}}{3.1}}

  (1)  By the definition \eqref{eqnTheta} of $\TsetR$ and Lemma~\ref{l:lemmaFromAistats}, $\forall\theta\in\TsetR, x \in\Xset$, it
holds that
\[
\int_{V_\theta} q_\theta(x,y) \,\mathrm{d}y = \sum
_{P \in\P} \int_{V_\theta} \mathcal{N}( P y |
x, c \Sigma) \,\mathrm{d}y =1 .
\]

(2)  Let $(X_t)_{t\geq0}$ and
$(\theta_t)_{t\geq0}$ be the random processes defined by
Algorithm \ref{figPseudocodeStableAMOR}. Let
$\cF_{t}=\sigma(X_0,\theta_0,\dots,X_{t},\theta_{t})$. We prove that
for any measurable positive function $f$,
\[
\mathbb{E} \bigl[f(X_t)|\cF_{t-1} \bigr] = \int
f(x_t)P_{\theta_{t-1}}(X_{t-1},x_t)
\,\mathrm{d}x_t ,\quad\quad \mbox{w.p.1.}
\]
Let $f$ be measurable and positive. Let $(\tilde P, \tilde X)$ be the r.v.
defined by steps 5 and 6. Let $U$ be a
uniform r.v. independent of $\sigma(X_0,\theta_0,\dots,X_{t-1},\theta_{t-1},
\tilde P, \tilde X)$. By construction, it holds that
%
\begin{eqnarray}\label{e:DecAccRej}
\mathbb{E} \bigl[f(X_t)|\cF_{t-1} \bigr] &=&
\mathbb{E} \bigl[f(\tilde P \tilde X) \bigl( 1- \alpha_{\theta_{t-1}}(X_{t-1},
\tilde P \tilde X) \bigr)|\cF_{t-1} \bigr]
\nonumber
\\[-8pt]\\[-8pt]
&&{} + f(X_{t-1}) \mathbb{E} \bigl[ \bigl( 1- \alpha_{\theta_{t-1}}(X_{t-1},
\tilde P \tilde X) \bigr)|\cF_{t-1} \bigr]  .\nonumber
\end{eqnarray}

Now note that the projection mechanism (steps 15 to
17 of Algorithm \ref{figPseudocodeStableAMOR}) guarantees
that $\theta_{t-1}\in\TsetR$ with probability 1. By Lemma~\ref{l:lemmaFromAistats}, $\theta\in\TsetR$ implies $\Xset=\bigcup
_{P}(PV_\theta)$
and
\[
\forall P,Q\in\cP\quad\mbox{such that}\quad P\neq Q, \Leb(PV_\theta\cap
QV_\theta)=0.
\]
Thus, for any measurable and bounded function
$\varphi\dvtx \Xset\times\TsetR\rightarrow\Rset$, we have
\[
\int_{\Xset} \varphi(x,\theta) \,\mathrm{d}x = \sum
_{Q\in\cP} \int_{QV_\theta\cap(\cup_{R\neq Q}RV_\theta)^c}
\varphi(x,\theta) \,\mathrm{d}x .
\]
Applying this decomposition to the second term in the RHS of
\eqref{e:DecAccRej} yields
\begin{eqnarray*}
&&\mathbb{E} \bigl[f(\tilde P \tilde X)  \mathbh{1}_{U \leq
\alpha_{\theta_{t-1}}(X_{t-1},\tilde P \tilde X)} |
\cF_{t-1} \bigr]
\\
&&\quad= \sum_{P\in\cP} \int h(Px)\frac{1}{N(x,\theta_{t-1})}
\mathbh{1}_{V_{\theta_{t-1}}}(Px)\cN(x| X_{t-1},c\Sigma_{t-1})
\,\mathrm{d}x
\\
&&\quad= \sum_{P,Q\in\cP} \int_{QV_{\theta_{t-1}}\cap(\bigcup_{R\neq Q}RV_{\theta
_{t-1}})^c} h(Px)
\frac{1}{N(x,\theta_{t-1})}\mathbh{1}_{V_{\theta_{t-1}}}(Px)\cN(x| X_{t-1},c
\Sigma_{t-1}) \,\mathrm{d}x,
\end{eqnarray*}
where $N(x,\theta) = |\{Q\in\cP/ Qx\in V_\theta\}|$. Using Lemma~\ref{l:lemmaFromAistats} again,
\[
\theta\in\TsetR, \quad\quad x\notin\bigcup_{P \neq Q} (PV_\theta\cap
QV_\theta) \Rightarrow N(x,\theta) = 1,
\]
and thus
\begin{eqnarray*}
\mathbb{E} \bigl[f(\tilde P \tilde X) \mathbh{1}_{U \leq
\alpha_{\theta_{t-1}}(X_{t-1},\tilde P \tilde X)} |
\cF_{t-1} \bigr] &=& \sum_{P\in\cP} \int h(y)
\mathbh{1}_{V_{\theta_{t-1}}}(y)\cN \bigl(P^{-1}y| X_{t-1},c
\Sigma_{t-1} \bigr) \,\mathrm{d}y
\\
&=& \int_{V_{\theta_{t-1}}} h(y)q_{\theta_{t-1}}(X_{t-1},y) \,\mathrm{d}y ,
\end{eqnarray*}
where in the last step we used the fact that $\cP$ is a group. Similarly,
\begin{eqnarray*}
&&\mathbb{E} \bigl[ \bigl( 1- \alpha_{\theta_{t-1}}(X_{t-1},\tilde P
\tilde X) \bigr)| X_0,\theta_0,\dots,X_{t-1},
\theta_{t-1} \bigr]
\\
&&\quad= \int_{V_{\theta_{t-1}}}  \bigl(1 - \alpha_{\theta_{t-1}}(X_{t-1},y)
\bigr) q_{\theta_{t-1}}(X_{t-1},y) \,\mathrm{d}y ;
\end{eqnarray*}
and this concludes the proof.

 (3)  This proof amounts to check the classical detailed balance
condition \cite{RoCa04}, and it is thus omitted. It is included in the
supplemental article \cite{6}.

\subsection{The Lyapunov function}
Lemma~\ref{l:lyapunov} establishes the existence of a Lyapunov function
for the mean field $h$ given by \eqref{def:h}.
%
\begin{lemm}
\label{l:lyapunov}
Under Assumption~\ref{a:target}, the mean field $h$ is continuous on
$\TsetR$,
the function $w$ defined by \eqref{e:defLyapunov} is $\cC^1$ on $\TsetR
$ and
\begin{enumerate}[(3)]
\item[(1)]$\nabla_\mu w(\param) = - \Sigma^{-1} h_\mu(\theta)$ and
$\nabla_\Sigma w(\param) = -\frac{1}{2} \Sigma^{-1} h_\Sigma(\theta)
\Sigma^{-1}$.
\item[(2)]$\langle\nabla w(\theta), h(\theta)\rangle\leq0$ on $\TsetR$
and $\langle\nabla w(\theta), h(\theta)\rangle= 0$ iff $\theta\in\cL$.
\item[(3)] For any $M>0$, the level set
%
\begin{equation}
\label{eq:def:levelset:lyap} \cW_M= \bigl\{\theta\in\TsetR\dvtx  w(\theta )\leq M
\bigr\}
\end{equation}
is a compact subset of $\TsetR$, and there exist $\delta_1,\delta_2>0$ such
that
\begin{subequations}
%
\begin{equation}\label{l:lyapunov:se:1}
\inf_{\theta\in\cW_M} \inf_{P\in\cP^*} \bigl\|(I-P)
\Sigma^{-1}\mu\bigr\|\geq\delta_1
\end{equation}
  and
\begin{equation}\label{l:lyapunov:se:2}
\inf_{\theta\in\cW_M} \lambda_{\min}(\Sigma) \geq
\delta_2,
\end{equation}
where $\lambda_{\min}(\Sigma)$ denotes the minimal eigenvalue of the real
symmetric matrix $\Sigma$.
\end{subequations}
%
\end{enumerate}
%
\end{lemm}
%
\begin{rem}
\label{r:VtoK}
As a consequence of Lemma~\ref{l:lyapunov}, observe that for any $M
>0$, there
exists $\delta>0$ such that $\cW_M \subseteq\TsetC_\delta$, where
$\TsetC_\delta$ is defined in \eqref{eqnKDelta}.
\end{rem}

\begin{pf}{\em(Continuity of $h$.)} This proof is a
straightforward application of Lebesgue's dominated convergence
theorem, and it is thus
omitted. It is included in the supplemental article, a link to which
can be found at the end of this paper \cite{6}.

{\em($w$ is $C^1$ on $\TsetR$.)} It is shown in \cite{BaCaFoKe12}, Proposition~3 of the
supplementary material, that the first term in the RHS of
\eqref{e:defLyapunov} is continuously differentiable on $\TsetR$. Since
$\|
(I-P)\Sigma^{-1}\mu\| \neq0$ for any $P \in\P^*$ and $(\mu, \Sigma)
\in
\TsetR$, the second term in the RHS of \eqref{e:defLyapunov} is continuously
differentiable on $\TsetR$. By \cite{BaCaFoKe12}, Proposition~3 of the supplementary
material, it holds for any $\param= (\mu, \Sigma) \in
\TsetR$
that
\begin{eqnarray*}
\nabla_\mu w(\theta) &=& -\Sigma^{-1}(\mu_{\pi_\theta} -
\mu) + \alpha\sum_P \frac{1}{\|(I-P)\Sigma^{-1}\mu\|^4}
\Sigma^{-1}U_P\Sigma^{-1}\mu
\\
&=&-
\Sigma^{-1}h_\mu(\theta),
\\
\nabla_\Sigma w(\theta) &=& -\frac{1}{2}\Sigma^{-1}
\bigl(\Sigma_{\pi_\theta} - \Sigma+ (\mu-\mup) (\mu-\mup)^T \bigr)
\Sigma^{-1}
\\
&&{} - \frac{\alpha}{2} \sum_P \frac{1}{\|
(I-P)\Sigma^{-1}\mu\|^4}
\Sigma^{-1} \bigl( \mu\mu^T\Sigma^{-1}U_P
\bigr)\Sigma^{-1} + U_P\Sigma^{-1}\mu
\mu^T
\\
&=& -\frac{1}{2} \Sigma^{-1} h_\Sigma(\theta)
\Sigma^{-1} .
\end{eqnarray*}
Hence, upon noting that $h_\Sigma(\theta)$ and $\Sigma^{-1}$ are symmetric,
\begin{eqnarray*}
\bigl\langle\nabla w(\theta),h(\theta) \bigr\rangle&=& -h_\mu(
\theta)^T\Sigma^{-1}h_\mu(\theta) -
\tfrac{1}{2}\operatorname{Trace} \bigl( \Sigma^{-1}
h_\Sigma(\theta) \Sigma^{-1} h_\Sigma(\theta) \bigr)
\\
&=& -h_\mu(\theta)^T\Sigma^{-1}h_\mu(
\theta) - \tfrac{1}{2}\operatorname{Trace} \bigl( \Sigma^{-1/2}
h_\Sigma(\theta) \Sigma^{-1} h_\Sigma(\theta)
\Sigma^{-1/2} \bigr).
\end{eqnarray*}
The first term of the RHS is negative since $\Sigma\in\cC_d^+$ and the second
term is negative since $(A,B)\mapsto\operatorname{Trace}(A^T B)$ is a
scalar product.
Therefore, $\langle\nabla w(\theta),h(\theta)\rangle\leq0$ with
equality if and only if f
$\theta\in\cL$.

{\em($\cW_M$ is compact.)} We prove \eqref{l:lyapunov:se:1}. By the definition
\eqref{e:defLyapunov} of $w$, for any $\theta\in\cW_M$, we have
\[
- \int\log\cN(x|\theta)\pi_\theta(x) \,\mathrm{d}x + \frac{\alpha}{2} \sum
_{P
\in\P^*} \frac{1}{\|(I-P)\Sigma^{-1}\mu\|^2}\leq M .
\]
In particular, the first term in the LHS is a cross-entropy, and it is
thus nonnegative (alternatively, see \cite{BaCaFoKe12}, Proposition~1 of the supplementary
material). Consequently, for any $\theta\in\cW_M$, we have
\[
\sum_{P\in\cP^*} \frac{1}{\|(I-P)\Sigma^{-1}\mu\|^2}\leq
\frac
{2M}{\alpha}.
\]
This yields $\|(I-P)\Sigma^{-1}\mu\|^2\geq\frac{\alpha}{2M}$ for
any $P\in\cP^*$, thus concluding the proof of \eqref{l:lyapunov:se:1}.

We now prove \eqref{l:lyapunov:se:2}. Let $\theta= (\mu, \Sigma) \in\cW_M$.
Denote by $(\lambda_i(\Sigma))_{i \leq d}$ the eigenvalues of $\Sigma$. Since
$\Sigma$ is symmetric, there exist $d \times d$ matrices $Q_\theta,
\Lambda_\theta$ such that $\Sigma= Q_\theta\Lambda_\theta
Q_\theta^T$, $Q_\theta$ is orthogonal, and
$\Lambda_\theta=\diag(\lambda_i(\Sigma))$. Then
%
\begin{eqnarray}
\label{eq:tool:Wcompact:1} 2M&\geq& 2w(\theta) \geq- 2\int\log\cN (x|\theta)
\pi_\theta(x) \,\mathrm{d}x
\nonumber
\\
& = & d\log(2\uppi) + \log\det\Sigma+ (\mup-\mu)^T
\Sigma^{-1}(\mup-\mu) + \operatorname{Trace} \bigl(\Sigma^{-1}
\Sigmap \bigr)
\\
&\geq& \sum_{i=1}^d \log
\lambda_{i}(\theta) + 0 + \operatorname{Trace} \bigl(
\Sigma^{-1}\Sigmap \bigr)
\nonumber
. 
\end{eqnarray}
Set $b_{i}(\theta) = (Q_\theta^T\Sigmap Q_\theta)_{ii}$. Then
%
\begin{equation}
\label{eq:tool:Wcompact:2} \operatorname{Trace} \bigl(\Sigma^{-1} \Sigmap \bigr) =
\operatorname{Trace} \bigl(Q_\theta\Lambda_\theta^{-1}Q_\theta^T
\Sigmap \bigr) = \operatorname{Trace} \bigl(Q_\theta^T
\Sigmap Q_\theta\Lambda_\theta^{-1} \bigr) = \sum
_{i=1}^d \frac{b_i(\theta)}{\lambda_{i}(\theta)}.
\end{equation}
Therefore, for any $\param\in\cW_M$,
%
\begin{equation}
\label{e:minorant} \sum_{i=1}^d \log
\lambda_{i}(\theta)+ \frac{b_i(\theta)}{\lambda_{i}(\theta)} \leq2 M .
\end{equation}

We now prove that for any $i$, $\inf_{\cW_M}b_i>0$. This property, combined
with \eqref{e:minorant}, will conclude the proof of \eqref{l:lyapunov:se:2}.
Let $\eps>0$ be such that $ 2^d \eps\|\pi\|_\infty\Delta_\pi^{d-1}<
|\cP|$, and for $v\in\{x\in\Rset^d \dvtx  \| x\|= 1\}$,
let
%
\begin{equation}
\label{eqnBeps} B_\eps^v(\theta)= \bigl\{x\in\Supp(\pi)
\cap V_\theta\dvtx  \bigl|\langle x-\mup,v\rangle\bigr|\leq\eps \bigr\}.
\end{equation}
Note that by Assumption~\ref{a:target},
\[
\pi\bigl(B_\eps^v(\theta) \bigr) \leq\|\pi
\|_\infty \Leb\bigl(B_\eps^v(\theta) \bigr)
\leq2^d\eps\|\pi\|_\infty \Delta_\pi^{d-1}.
\]
Then, by definition of $\eps$,
%
\begin{equation}\label{e:borneBande}
\pi \bigl(V_\theta\setminus B_\eps^v(\theta)
\bigr) \geq|\cP|- 2^d\eps\|\pi\|_\infty
\Delta_\pi^{d-1}>0.
\end{equation}
Now, if $(e_i)$ denotes the canonical basis of $\mathbb{R}^d$, then
%
\begin{eqnarray}
\label{eq:tool:Wcompact:3} b_i(\theta) &=& |\cP| e_i^T
Q_\param^T \biggl(\int_{V_\theta} (x-\mup)
(x-\mup)^T\pi(x) \,\mathrm{d}x \biggr) Q_\param e_i
\nonumber
\\
&=& |\cP|\int_{V_\theta} (Q_\param
e_i)^T(x-\mup) (x-\mup)^T Q_\param
e_i \pi(x) \,\mathrm{d}x
\nonumber
\\
&=& |\cP|\int_{V_\theta} \langle x-\mup,
Q_\param e_i\rangle^2\pi(x) \,\mathrm{d}x
 \\
&\geq& |\cP|\int_{V_\theta\setminus B_{\eps}^{Q_\theta
e_i}(\theta)} \langle x-\mup,
Q_\param e_i\rangle^2\pi(x) \,\mathrm{d}x
\nonumber
\\
&\geq& \eps^2 |\cP|\pi \bigl(V_\theta\setminus
B_\eps^{Q_\param e_i}(\theta) \bigr),\nonumber
\end{eqnarray}
where the last inequality follows from the definition \eqref{eqnBeps} of
$B_\eps^{Q_\param e_i}(\theta)$. Thus, by \eqref{e:borneBande},
$b_i(\theta)$
is bounded away from zero on $\cW_M$.

As $w$ is continuous on $\TsetR$, $\{\theta\in\TsetR, w(\theta)\leq M\}
$ is
closed. From \eqref{l:lyapunov:se:2}, \eqref{eq:tool:Wcompact:1} and
Assumption~\ref{a:target}, $\mu\mapsto(\mup-\mu)^T\Sigma^{-1}(\mup-\mu
)$ is
bounded on $\cW_M$. In addition, \eqref{eq:tool:Wcompact:1},
\eqref{eq:tool:Wcompact:2} and \eqref{eq:tool:Wcompact:3} imply that
$\Sigma
\mapsto\log\det\Sigma$ is bounded on $\cW_M$. These properties
combined with
\eqref{l:lyapunov:se:2} imply that $\cW_M$ is bounded. Hence, $\cW_M$ is
compact.
\end{pf}

\subsection{Regularity in \texorpdfstring{$\param$}{} of the Poisson solution}
%
\begin{lemm}
\label{l:poisson:bis}
\begin{enumerate}[(2)]
\item[(1)]  For any $M>0$, there exists $\rho\in(0,1)$
such that for any $x \in\Xset$ and any $\theta\in\cW_M$, $\|
P_\param^n(x,\cdot) - \pi_\param\|_{\operatorname{TV}} \leq2(1-\rho)^n$.
\item[(2)] Under Assumption~\ref{a:target}, for
any $\param\in\TsetR$, there
exists a solution $\hat{H}_\param$ of the Poisson equation, that is, $
\hat{H}_\param-P_{\param}\hat{H}_\param=
H(\cdot,\param) - \pi_{\theta}H(\cdot,\param)$. Furthermore, for any
$M>0$,
%
\begin{equation}\label{e:poissonBounded}
\sup_{\theta\in\cW_M}\sup_{x\in\Xset} \bigl|
\hat{H}_\param(x) \bigr|< \infty.
\end{equation}
\end{enumerate}
\end{lemm}

\begin{pf} (1) It is
sufficient to
prove that there exists $\rho\in(0,1)$ such that for any $x \in\Xset$ and
$\param\in\cW_M$, $P_\param(x, \cdot) \geq\rho\pi_\param$ (see, e.g.,
\cite{MeTw93}, Theorem~16.2.4). By \eqref{e:kernelP}, for any $x\in\Xset
$ and
$A \in\Xfield$, $P_{\param}(x,A) \geq\int_{A \cap\Vt} \alpha_{\param}(x,y)
q_{\param}(x,y) \,\mathrm{d}y$. By Lemma~\ref{l:lyapunov}, there exists $a>0$ such
that for any $(\mu,\Sigma) \in\cW_M$, any $m,z \in\Xset$, and any $P
\in
\P$, we have $\mathcal{N}(P z | m, \Sigma) \geq a$. Thus, for any
$\param\in\cW_M$ and $y \in V_\param$, it holds that
%
\begin{equation}\label{e:kernelMinorant}
\alpha_{\param}(x, y) q_{\param}(x,y)\mathbh{1}_{\Vt}(y)
\geq a |\P| \biggl( 1 \wedge\frac{\pi(y)}{\pi(x)} \biggr) \mathbh
{1}_{\Vt}(y) \geq\frac{a}{\|\pi\|_\infty} \pi_\param(y) .
\end{equation}
Thus, we have $P_{\param}(x,\cdot) \geq\rho\pi_\param$ for any $x \in
\Xset$
and $\param\in\cW_M$ with $\rho= a / \|\pi\|_\infty$.

 (2)
By item (1),
\[
\hat{H}_\param(x) = \sum_n
P_{\param}^n \bigl(H(x,\param)-\pi_{\param}\bigl(H(
\cdot,\param)\bigr) \bigr)
\]
exists and solves the Poisson equation. \eqref{e:poissonBounded}
trivially follows from item (1).
\end{pf}

\begin{lemm}\label{l:voronoiControl}
Let $M>0$ and $\kappa\in(0,1/2)$. Under
Assumption~\ref{a:target}, there exists
$C>0$ such that for any
$\theta\in\cW_M$ and $\theta'\in\TsetR$, it holds that
%
\begin{equation}
\Leb(V_\theta\setminus V_{\theta'}) \leq C \bigl\|\theta-
\theta'\bigr\|^{1-2\kappa}, \label{e:lebLipschitz}
\end{equation}
where $\Leb(A)$ denotes the Lebesgue measure of the set $A$.
\end{lemm}

\begin{pf}
We prove that there exist $\bar C, \bar h >0$, such that for any
$\theta\in
\cW_M$ and any $\theta' \in\Theta$ such that $\|\theta- \theta'
\|
\leq\bar h$, $\Leb(V_\theta\setminus V_{\theta'}) \leq\bar C
\|\theta-\theta'\|^{1-2\kappa}$. Note that since $V_\theta\subset
\Xset$ and since $\Xset$ is bounded, there exists $\check C>0$ such that
$\Leb(V_\theta\setminus V_{\theta'}) \leq\check C$. Therefore,
\eqref{e:lebLipschitz} holds with $C = \bar C \vee\check C / \bar
h^{1- 2
\kappa}$.

By Lemma~\ref{l:lyapunov}, $w$ is uniformly continuous on
$\cW_{M+1}$, and there exists $h_0>0$ small enough for which
%
\begin{equation}
\bigl[ \theta\in\cW_M, \theta'\in\Theta, \bigl\|\theta-
\theta'\bigr\|<h_0 \bigr] \Rightarrow\forall u\in[0,1],\quad\quad
\theta+u \bigl(\theta'-\theta \bigr)\in\cW_{M+1}.
\label{e:unifCont}
\end{equation}
Let $\bar h \leq h_0$. Let $\theta=(\mu,\Sigma)\in\cW_M$ and $\theta'
\neq
\theta$ such that $\|\theta- \theta' \|\leq\bar h$.

By definition of the set $V_\vartheta$, for any $x\in\Vt\setminus
V_{\theta'}$,
there exists $P\in\cP^*$ such that $L_{\theta'}(x)-L_{\theta'}(P^Tx)>0$ and
$L_{\theta}(x) - L_{\theta}(P^Tx)\leq0$. Since $\vartheta\mapsto
L_\vartheta(x) - L_\vartheta(P^T x)$ is continuous on $\cW_{M+1}$,
there exists
$u\in[0,1]$ depending on $x,\theta,\theta'$, and $P$ such that
$L_{\theta+u(\theta'-\theta)}(x) - L_{\theta+u(\theta'-\theta)}(P^Tx)=0$.
Therefore,
\[
\Vt\setminus V_{\theta'} \subset\bigcup_{P\in\cP^*}
\mathcal{V}_P,
\]
where
%
\begin{equation}\label{e:zeros}
\mathcal{V}_P=\bigcup_{u\in[0,1]} \cZ
\bigl(L_{\theta+u(\theta'-\theta)}(\cdot) - L_{\theta+u(\theta'-\theta
)} \bigl(P^T\cdot
\bigr) \bigr) \cap\Xset;
\end{equation}
and $\cZ(f)$ denotes the zeros of the function $f$. The proof proceeds by
showing that for any $P\in\cP^*$, $\mathcal{V}_P$ is included in a measurable
set with measure $O ( \|\theta-\theta'\|^{1-2\kappa} )$.

Let $P \in\cP^*$. Let $B(0, \Delta_\pi) = \{y \in\Rset^d \dvtx  \| y
\|\leq
\Delta_\pi\}$, where $\Delta_\pi$ is defined by \ref{eq:def:DeltaPi}.
For any
$x\in B(0, \Delta_\pi)$, define
\begin{eqnarray*}
l_\theta(x)& =&2\mu^T\Sigma^{-1}
\bigl(I-P^T \bigr)x ,
\\
q_\theta(x)&=&x^T \bigl(\Sigma^{-1}-P
\Sigma^{-1}P^T \bigr)x ,
\\
\mathsf{B}_{\theta,
\theta'} &=& \bigl\{ x\in B(0,\Delta_\pi) \dvtx  \bigl|
l_\theta(x) \bigr|\leq\bigl\|\theta-\theta'
\bigr\|^{\kappa} \bigr\}.
\end{eqnarray*}
Denote by $\Sset$ the unit sphere $\{x\in\Rset^d / \| x\|=1\}$. Let
$u \in[0,1]$ and $tv \in\cZ(L_{\theta+u(\theta'-\theta)}(\cdot) -
L_{\theta+u(\theta'-\theta)}(P^T\cdot) ) \cap\Xset$ where $t\in[0,
\Delta_\pi]$ and $v \in\Sset$. Upon noting that for any $\vartheta
\in\cW_{M+1}$,
%
\begin{equation}
L_{\vartheta}(tv) - L_{\vartheta} \bigl(tP^Tv \bigr) = t
\bigl( q_\vartheta(v)t -l_\vartheta(v) \bigr) \label{e:polynome},
\end{equation}
we consider several cases:
\begin{enumerate}[(iii)]
\item[(i)]  $tv\in\mathsf{B}_{\theta
,\theta'}$.
\item[(ii)]  $tv \notin\mathsf{B}_{\theta
,\theta'}$
and $q_{\theta+u(\theta'-\theta)}(v)=0$. Then, by \eqref{e:polynome},
$l_{\theta+ u(\theta'-\theta)}(tv) =0$ which implies that $tv\in
\mathsf{B}_{\theta,\theta'}$. This yields a contradiction.
\item[(iii)]  $tv\notin\mathsf{B}_{\theta
,\theta'}$ and
$q_{\theta+u(\theta'-\theta)}(v)\neq0$. Then $t \neq0$ and, by \eqref
{e:polynome},
%
\begin{equation}
t = \frac{l_{\theta+u(\theta'-\theta)}(v)}{q_{\theta+u(\theta'-\theta
)}(v)} . \label{eq:definitionOft}
\end{equation}
Since we assumed $t\in[0,\Delta_\pi]$, this ratio is positive. In
order to
characterize the point $tv$, additional notations are required. First, note
that by Lemma~\ref{l:lyapunov}, there exists $C_1>0$ such that for any
$\tilde
\theta=(\tilde\mu, \tilde\Sigma) \in\cW_{M+1}$,
\[
\|\tilde\theta- \theta\|\leq h_0 \Rightarrow\bigl\|\tilde
\Sigma^{-1}-\Sigma^{-1}\bigr\|\leq C_1 \|\tilde
\Sigma-\Sigma\|.
\]
Thus, there exists $C_2>0$ such that for any $\tilde\theta\in\cW_{M+1}$,
$\|\tilde\theta-\theta\|\leq h_0$, and for any $x\in B(0,\Delta
_\pi)$,
%
\begin{eqnarray}\label{l:voronoiControl:tool1}
\bigl| l_{\tilde\theta}(x)-l_\theta(x)\bigr|&=& 2 \bigl\llvert
\mu^T \bigl[\tilde\Sigma^{-1}-\Sigma^{-1} \bigr]
\bigl(I-P^T \bigr)x +(\tilde\mu-\mu)^T \tilde
\Sigma^{-1} \bigl(I-P^T \bigr)x \bigr\rrvert\quad\quad
\nonumber
\\[-8pt]\\[-8pt]
&\leq& C_2 \|\tilde\theta-\theta\|. \nonumber
\end{eqnarray}
Note that since $x,\mu\in B(0,\Delta_\pi)$, $C_2$ does not depend on
$x$ and
$\theta$. Similarly, there exists $C_3>0$ such that for $x\in B(0,\Delta
_\pi)$
and $\tilde\theta\in\cW_{M+1}$ satisfying $\|\tilde
\theta-\theta\|\leq h_0$,
%
\begin{equation}
\label{eq:l:voronoiControl:tool2} \bigl| q_{\tilde\theta}(x) - q_\theta (x)\bigr|\leq
C_3 \|\tilde\theta- \theta\|.
\end{equation}
We can assume without loss of generality that $ \bar h$ is small enough
so that
%
\begin{equation}
\label{eq:l:voronoiControl:tool3}
\bigl\|\theta-\theta'\bigr\|\leq \bar h \Rightarrow
\bigl\|\theta-\theta'\bigr\|^{\kappa} - ( C_2 + 2
C_3 \Delta_\pi) \bigl\|\theta- \theta'\bigr\|
\geq\tfrac{1}{2} \bigl\| \theta- \theta'\bigr\|^{\kappa} .
\end{equation}
We now distinguish three subcases.
\begin{enumerate}[(c)]
\item[(a)]  $v\in\mathsf{B}_{\theta,\theta'}$.
\item[(b)]  $v\notin\mathsf{B}_{\theta
,\theta'}$
and $q_\theta(v)\neq0$. Since $t \in[0,\Delta_\pi]$,
\eqref{eq:definitionOft} implies that
$| q_{\theta+u(\theta'-\theta)}(v)|\geq
| l_{\theta+u(\theta'-\theta)}(v)|/\Delta_\pi$. Since $v \notin
\mathsf{B}_{\theta,\theta'}$, $| l_\theta(v) |\geq\|
\theta-\theta'\|^{\kappa}$ and by using \eqref{l:voronoiControl:tool1},
\[
| l_{\theta+u(\theta'-\theta)}|\geq\bigl| l_\theta(v)\bigr|- \bigl\llvert
l_{\theta+u(\theta'-\theta)}- l_\theta(v) \bigr\rrvert\geq\bigl\|\theta-
\theta'\bigr\|^{\kappa}-C_2\bigl\|\theta-
\theta'\bigr\|.
\]
Hence, it holds that $| q_{\theta+u(\theta'-\theta)}(v)|\geq
(\|
\theta-\theta'\|^{\kappa}-C_2\|\theta-\theta'\|)/\Delta_\pi$,
and, by
\eqref{eq:l:voronoiControl:tool2}, we have $| q_\theta(v)| \geq
| q_{\theta+u(\theta'-\theta)}(v)| - C_3 \| \theta- \theta'\|
$. These inequalities
together with \eqref{l:voronoiControl:tool1} and \eqref
{eq:l:voronoiControl:tool3}
lead to
\[
\biggl\llvert t - \frac{l_\theta(v)}{q_\theta(v)}\biggr\rrvert = \biggl\llvert
\frac{l_{\theta+u(\theta'-\theta)}(v)}{q_{\theta+u(\theta'-\theta)}(v)} -
 \frac{l_\theta(v)}{q_\theta(v)}\biggr\rrvert \leq C_4\bigl\|
\theta-\theta'\bigr\|^{1-2\kappa},
\]
for some $C_4>0$.
\item[(c)]  $v\notin\mathsf{
B}_{\theta,\theta'}$ and $q_{\theta}(v)= 0$. Then by
\eqref{l:voronoiControl:tool1} and \eqref{eq:l:voronoiControl:tool2},
\[
t \geq\frac{\|
\theta-\theta'\|^\kappa-C_2 \|
\theta-\theta'\|}{C_3\|
\theta-\theta'\|} \geq2 \Delta_\pi,
\]
which is in contradiction with the assumption that $t \leq\Delta_\pi$.
\end{enumerate}
\end{enumerate}

As a conclusion, we have just proved that $\mathcal{V}_P$ is included
in the
union of three sets defined by $\mathsf{B}_{\theta,\theta'}$
(case (i)), by $\{tv \dvtx  t\in
[0,\Delta_\pi],v\in\Sset\cap\mathsf{B}_{\theta,\theta'}\}$
(case (iii)(a)), and by
\[
\biggl\{tv : v\in\Sset, v\notin\mathsf{B}_{\theta,\theta'}, q_\theta(v)
\neq0, 0 \leq t \leq\Delta_\pi, \biggl\llvert t - \frac{l_\theta(v)}{q_\theta(v)}
\biggr\rrvert\leq C_4\bigl\|\theta- \theta'
\bigr\|^{1-2\kappa} \biggr\}
\]
(case (iii)(c)). This concludes the first step.\def\thefigure{7}\vadjust{\goodbreak}
%
\begin{figure}[b]

\includegraphics{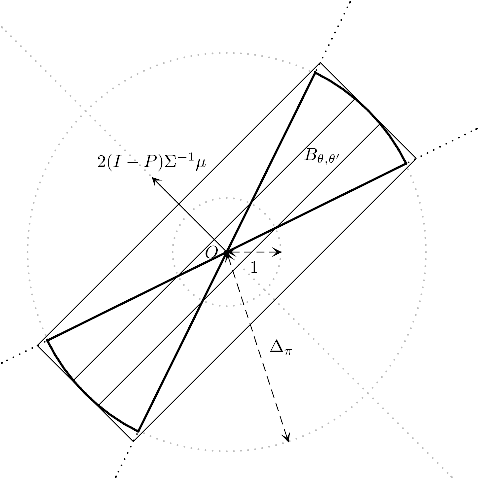}

\caption{Bounding the measure of the set $\cV_P^{(2)}$.}
\label{f:band}
\end{figure}

The second step consists in computing an upper bound for the Lebesgue
measure of
each of these three sets. For simplifying the presentation, we detail
the case
$d=2$ and use polar coordinates $(\rho,\phi)$; the argument remains
valid when
$d>2$ using generalized spherical coordinates. Define $t_\theta(\phi)=
l_\theta(\mathrm{e}^{\mathrm{i}\phi})/q_\theta(\mathrm{e}^{\mathrm{i}\phi})$. Rephrasing the conclusion of
the first
step, we have $\mathcal{V}_P \subset\bigcup_{\ell=1}^3 \mathcal
{V}_P^{(\ell)}$
with
\begin{eqnarray*}
\mathcal{V}_P^{(1)} &=& \mathsf{B}_{\theta,\theta'} ,
\\
\mathcal{V}_P^{(2)} &=& \bigl\{ (\rho,\phi) / \rho\in[0,
\Delta_\pi], \mathrm{e}^{\mathrm{i}\phi}\in\mathsf{B}_{\theta,\theta'} \bigr\} ,
\\
\mathcal{V}_P^{(3)} &=& \bigl\{(\rho,\phi) /
\mathrm{e}^{\mathrm{i}\phi}\notin\mathsf{B}_{\theta,\theta'}, q_\theta
\bigl(\mathrm{e}^{\mathrm{i}\phi} \bigr)\neq0, 0 \leq\rho\leq\Delta_\pi, \bigl
\llvert\rho- t_\theta(\phi) \bigr\rrvert\leq C_4\bigl\| \theta-
\theta'\bigr\|^{1-2\kappa} \bigr\} .
\end{eqnarray*}
\normalsize These sets are Borel sets. By definition of $\cW_M$,
$l_\theta$ is not identically
zero, and thus
\[
\Leb \bigl(\mathcal{V}_P^{(1)} \bigr) = \Leb(
\mathsf{B}_{\theta,\theta'})\leq2 \Delta_\pi\frac{\|\theta-\theta
'\|^{1-2\kappa}}{\| 2 \mu^t \Sigma^{-1} (I - P^T)
\|} \leq
C_5 \bigl\|\theta-\theta'\bigr\|^{1-2\kappa}
\]
for some $C_5>0$ as a consequence of Lemma~\ref{l:lyapunov}. For
$\mathcal{V}_P^{(2)}$, note that it is upper bounded by the reunion of
the two
circular sectors in bold lines in Figure~\ref{f:band}. This area is easily
bounded by the area of the outer rectangle, which is proportional to
$\|
\theta-\theta'\|^{1-2\kappa}$. Finally,
\[
\Leb \bigl(\mathcal{V}_P^{(3)} \bigr) = \int
_0^{2\uppi} \biggl[\frac{\rho^2}{2}
\biggr]_{0\vee(t_\theta(\phi)-C_4\|
\theta-\theta'\|^{1-2\kappa})} ^{\Delta_\pi\wedge
(t_\theta(\phi) + C_4\|\theta-\theta'\|^{1-2\kappa})}
\mathbh {1}_{q_\theta(\mathrm{e}^{\mathrm{i} \phi}) \neq0} \,\mathrm{d}\phi.
\]
We can assume without loss of generality that $\bar h$ is small enough
so that
$2 C_4 \bar h^{1-2\kappa} < \Delta_\pi$. Therefore, we can partition
$[0, 2\uppi]
=\cA\cup\cB\cup\cC$, where
\begin{eqnarray*}
\cA&=& \bigl\{ \phi\in[0,2\uppi] / t_\theta(\phi) - C_4 \bigl\|
\theta- \theta'\bigr\|^{1-2\kappa} \geq0 \mbox{ and }
t_\theta(\phi) + C_4\bigl\|\theta- \theta'
\bigr\|^{1-2\kappa} \leq\Delta_\pi \bigr\},
\\
\cB&=& \bigl\{ \phi\in[0,2\uppi] / t_\theta(\phi) - C_4 \bigl\|
\theta- \theta'\bigr\|^{1-2\kappa} \geq0 \mbox{ and }
t_\theta(\phi) + C_4\bigl\|\theta- \theta'
\bigr\|^{1-2\kappa} \geq\Delta_\pi \bigr\},
\\
\cC&=& \bigl\{ \phi\in[0,2\uppi] / t_\theta(\phi) - C_4 \bigl\|
\theta- \theta'\bigr\|^{1-2\kappa} \leq0 \mbox{ and } 0 \leq
t_\theta(\phi) + C_4\bigl\|\theta- \theta'
\bigr\|^{1-2\kappa} \leq\Delta_\pi \bigr\}.
\end{eqnarray*}
This yields
%
\begin{eqnarray}
\hspace*{-8pt}\Leb \bigl(\mathcal{V}_P^{(3)} \bigr) &\leq&
2C_4 \int_\cA t_\theta(\phi) \bigl\|
\theta-\theta'\bigr\|^{1-2\kappa} \,\mathrm{d}\phi+ \frac{1}{2}\int
_\cB \bigl(\Delta_\pi^2 -
\bigl(t_\theta(\phi)- C_4\bigl\|\theta-\theta'
\bigr\|^{1-2\kappa} \bigr)^2 \bigr) \,\mathrm{d}\phi\hspace*{8pt}
\nonumber
\\[-8pt]\\[-8pt]
&&{}+ \frac{1}{2}\int_\cC \bigl(t_\theta(
\phi) + C_4\bigl\|\theta-\theta'\bigr\|^{1-2\kappa}
\bigr)^2 \,\mathrm{d}\phi\nonumber
\\
&\leq& C_6 \bigl\|\theta-\theta'\bigr\|^{1-2\kappa} ,
\label{e:measure}
\end{eqnarray}
for some $C_6>0$, since on $\cA$, $0 \leq t_\theta(\phi) \leq\Delta_\pi
$, on
$\cB$, $(t_\theta(\phi) - C_4\|\theta-\theta'\|^{1-2\kappa})^2
\geq
(\Delta_\pi- 2 C_4\|\theta-\theta'\|^{1-2\kappa})^2$, and on
$\cC$,
$|t_\theta(\phi)| \leq C_4\|\theta-\theta'\|^{1-2\kappa}$.

This concludes the proof.
\end{pf}

\begin{lemm}[(Regularity in $\param$ of the invariant distribution $\pi
_\param$)]\label{l:Regularity:theta:pi}
Let $M>0$ and $\kappa\in(0,1/2)$. Under Assumption~\ref{a:target}, there
exists $C>0$ such that for any $\theta\in\cW_M$ and $\theta'\in\TsetR$,
\[
\|\pi_{\param} - \pi_{\param'} \|_{\TV}\leq C\bigl\|
\theta-\theta'\bigr\|^{1-2\kappa}.
\]
\end{lemm}

\begin{pf}
By definition of the total variation,
\[
\|\pi_\theta- \pi_{\theta'}\|_{\TV} \leq| \cP
| \bigl(\pi(V_\theta\setminus V_{\theta'}) +
\pi(V_{\theta'}\setminus V_{\theta}) \bigr).
\]
Since
\[
V_{\theta'} \setminus V_{\theta} = V_\theta\setminus
(V_{\theta}\cap V_{\theta'} ),\quad\quad V_\theta\setminus
V_{\theta'} = V_{\theta} \setminus (V_{\theta}\cap
V_{\theta'} ),
\]
it holds that
\[
\pi(V_{\theta'}\setminus V_{\theta}) = \frac{1}{|\cP|} - \pi(
V_{\theta}\cap V_{\theta'}) = \pi(V_\theta\setminus
V_{\theta'}),
\]
where
we used Lemma~\ref{l:lemmaFromAistats}. Then, by Assumption~\ref
{a:target} and Lemma~\ref{l:voronoiControl},
there exists $C>0$ such that for any $\theta\in\cW_M$ and $\theta'\in
\TsetR$,
\[
\|\pi_\theta- \pi_{\theta'}\|_{\TV} \leq2\|\pi
\|_\infty \Leb(V_{\theta}\setminus V_{\theta'}) \leq C
\bigl\| \theta-\theta'\bigr\|^{1-2\kappa}.
\]
\upqed\end{pf}

\begin{lemm}[(Regularity in $\param$ of the kernels $P_\param$)]\label{l:Regularity:theta:kernel}
Let $M>0$ and $\kappa\in(0,1/2)$. Under Assumption~\ref{a:target}, there
exists $C>0$ such that for any $\theta\in\cW_M$ and $\theta'\in\cW_{M+1}$,
\[
\bigl\| P_{\param}(x,\cdot) - P_{\param'}(x,\cdot)
\bigr\|_{\TV}\leq C\bigl\|\theta-\theta'\bigr\|^{1-2\kappa}.
\]
\end{lemm}

\begin{pf}
From the
definition of the transition kernel $P_{\param}$, we have
%
\begin{eqnarray}\label{e:sumDeltas}
&&\bigl| P_{\param}f(x) - P_{\param'}f(x)\bigr|\nonumber
\\
&&\quad\leq \biggl\llvert
\int f(y) \bigl( \alpha_{\param}(x,y)q_{\theta}(x, y)
\mathbh{1}_{V_{\theta}}(y) - \alpha_{\param'}(x,y)q_{\theta'}(x, y)
\mathbh{1}_{V_{\theta'}}(y) \bigr) \,\mathrm{d}y \biggr\rrvert
\nonumber
\\
&&\quad\quad{} + \bigl| f(x) \bigr| \biggl\llvert\int \bigl( \alpha_{\param
'}(x,y)q_{\theta'}(x,
y)\mathbh{1}_{V_{\theta'}}(y) - \alpha_{\param}(x,y)q_{\theta}(x,
y)\mathbh{1}_{V_{\theta}}(y) \bigr) \,\mathrm{d}y \biggr\rrvert
\\
&&\quad\leq 2\| f \|_{\infty} \int\bigl|\alpha_{\param}(x,y)q_{\theta}(x,
y)\mathbh{1}_{V_{\theta}}(y) - \alpha_{\param'}(x,y)q_{\theta'}(x,
y)\mathbh{1}_{V_{\theta'}}(y) \bigr| \,\mathrm{d}y
\nonumber
\\
&&\quad= 2\| f \|_{\infty} \sum_{i=1}^4
\Delta^i_{\param,\param'}(x) ,\nonumber
\end{eqnarray}
where
\begin{eqnarray*}
\Delta^1_{\param,\param'}(x) &=& \int_{\cA_\param(x)\cap\cA_{\param
'}(x)} \bigl|
\alpha_{\param}(x,y)q_{\theta}(x, y)\mathbh{1}_{V_{\theta}}(y) -
\alpha_{\param'}(x,y)q_{\theta'}(x, y)\mathbh{1}_{V_{\theta'}}(y)
\bigr| \,\mathrm{d}y,
\\
\Delta^2_{\param,\param'}(x) &=& \int_{\cR_{\param}(x)\cap\cR_{\param
'}(x)} \bigl|
\alpha_{\param}(x,y)q_{\theta}(x, y)\mathbh{1}_{V_{\theta}}(y) -
\alpha_{\param'}(x,y)q_{\theta'}(x, y)\mathbh{1}_{V_{\theta'}}(y)
\bigr| \,\mathrm{d}y,
\\
\Delta^3_{\param,\param'}(x) &=& \int_{\cA_\param(x)\cap\cR_{\param
'}(x)} \bigl|
\alpha_{\param}(x,y)q_{\theta}(x, y)\mathbh{1}_{V_{\theta}}(y) -
\alpha_{\param'}(x,y)q_{\theta'}(x, y)\mathbh{1}_{V_{\theta'}}(y)
\bigr| \,\mathrm{d}y,
\\
\Delta^4_{\param,\param'}(x) &=& \int_{\cR_\param(x)\cap\cA_{\param
'}(x)} \bigl|
\alpha_{\param}(x,y)q_{\theta}(x, y)\mathbh{1}_{V_{\theta}}(y) -
\alpha_{\param'}(x,y)q_{\theta'}(x, y)\mathbh{1}_{V_{\theta'}}(y)
\bigr| \,\mathrm{d}y
\end{eqnarray*}
and
\[
\cA_\param(x)=\bigl\{y \dvtx  \alpha_\theta(x,y) = 1\bigr\},\quad\quad
\mathcal{R}_\param(x)=\bigl\{y \dvtx  \alpha_\theta(x,y)< 1\bigr\} .
\]
%
We now upper bound each term
%
\begin{eqnarray}\label{e:dec3}
\Delta^1_{\param,\param'}(x) &=& \int_{\cA_\theta(x)\cap\cA_{\theta
'}(x)} \biggl|
\sum_{Q\in\mathcal{P}} \bigl( \mathbh{1}_{V_\param}(y)
\mathcal{N}(Qy| x,\Sigma) - \mathbh{1}_{V_{\param'}}(y)\mathcal {N} \bigl(Qy
| x,\Sigma' \bigr) \bigr) \biggr| \,\mathrm{d}y
\nonumber
\\
&\leq& \int\bigl|\mathbh{1}_{V_\param}(y)-\mathbh{1}_{V_{\param'}}(y)
\bigr|\sum_{Q\in\mathcal{P}} \mathcal{N}(Qy| x,\Sigma)
\\
&&{}+ \mathbh{1}_{V_{\param'}}(y)\sum_{Q\in\cP} \bigl|
\mathcal{N}(Qy| x,\Sigma) - \mathcal{N} \bigl(Qy| x,\Sigma'
\bigr) \bigr| \,\mathrm{d}y.\nonumber
\end{eqnarray}
By Lemma~\ref{l:lyapunov}, there exist $a,b>0$ such that for any
$\param\in\cW_{M+1}$, $m,z \in\Xset$, and $Q \in\P$, we have
%
\begin{equation}\label{e:minMajGaussians}
a\leq\mathcal{N}(Q z | m, c \Sigma) \leq b ,
\end{equation}
so that the first term in the RHS of \eqref{e:dec3} is bounded by
\begin{eqnarray*}
\int \bigl\llvert\mathbh{1}_{V_\param}(y)-\mathbh{1}_{V_{\param'}}(y)
\bigr\rrvert\sum_{Q\in\mathcal{P}} \mathcal{N}(Qy| x,\Sigma)
\,\mathrm{d}y &\leq& |\cP| b\int \bigl\llvert\mathbh{1}_{V_\param}(y)-
\mathbh{1}_{V_{\param'}}(y) \bigr\rrvert \,\mathrm{d}y
\\
&=& |\cP| b \int \bigl(\mathbh{1}_{V_\theta\setminus V_{\theta
'}}(y) +
\mathbh{1}_{V_{\theta'}\setminus V_{\theta}}(y) \bigr) \,\mathrm{d}y
\\
&\leq& C \bigl\|\theta-\theta' \bigr\|^{1-2\kappa},
\end{eqnarray*}
where we used Lemma~\ref{l:voronoiControl}. Let us now consider the
second term
of the right-hand side of \eqref{e:dec3}. Using the uniform continuity
of $w$
on $\cW_{M+1}$ (see Lemma~\ref{l:lyapunov}), there exists $\bar{h}$ small
enough such that
%
\begin{equation}\label{e:segmentSecurity}
\theta\in\cW_M, \quad\quad\| h\|<\bar{h}\Rightarrow\theta+h\in
\cW_{M+1}.
\end{equation}
For any $\theta\in\cW_M$, $\theta'\in\cW_{M+1}$ such that $\|\theta-
\theta' \|\geq\bar h$, there exists $C_1$ such that
\[
\sum_{Q\in\cP} \bigl|\mathcal{N}(Qy| x,\Sigma) -
\mathcal{N}\bigl(Qy| x,\Sigma'\bigr)\bigr |\, \mathrm{d}y \leq
C_1 \bigl\|\theta- \theta'\bigr\|^{1- 2\kappa} .
\]
Assume now that $\theta\in\cW_M$, $\theta'\in\cW_{M+1}$ and
$\|\theta-\theta'\|<\bar{h}$. Denote by
%
\begin{equation}
\Sigma_t=(1-t)\Sigma+t\Sigma'. \label{e:Sigma_t}
\end{equation}
By \eqref{e:segmentSecurity} and \eqref{l:lyapunov:se:2}, $\Sigma_t^{-1}$
exists and $\sup_{t \leq1, \theta\in\cW_M, \theta' \in\cW_{M+1}}
\|
\Sigma_t^{-1} \|< \infty$. We can then write
%
\begin{eqnarray}\label{e:dec3term2}
\bigl|\mathcal{N}(Qy| x,\Sigma) - \mathcal{N} \bigl(Qy| x,
\Sigma' \bigr) \bigr|&=& \int_0^1
\mathcal{N}(Qy| x,\Sigma_t) \biggl\llvert\frac{ \mathrm{d}}{ \mathrm{d}t}\log
\mathcal{N}(Qy| x,\Sigma_t) \biggr\rrvert \,\mathrm{d}t
\nonumber
\\[-8pt]\\[-8pt]
&\leq& b\int_0^1 \biggl\llvert
\frac{ \mathrm{d}}{ \mathrm{d}t}\log\mathcal{N}(Qy| x,\Sigma_t) \biggr\rrvert
\,\mathrm{d}t.\nonumber
\end{eqnarray}
In addition, by Assumption~\ref{a:target}, there exists $C_2$ such that
%
\begin{equation}\label{e:gaussianLipschitz}
\biggl\llvert\frac{  \mathrm{d}}{ \mathrm{d}t}\log\mathcal{N}(Qy| x,\Sigma_t)
\biggr\rrvert= \bigl|(x-Qy)^T \Sigma_t^{-1}
\bigl(\Sigma'-\Sigma \bigr)\Sigma_t^{-1}(x-Qy)
\bigr|\leq C_2 \bigl\|\param-\param'\bigr\| .
\end{equation}
We thus have proved that
\[
\bigl[\theta\in\cW_M, \theta'\in\cW_{M+1},
\bigl\| \theta-\theta'\bigr\|<\bar{h} \bigr]\Rightarrow \bigl|
\mathcal{N}(Qy| x,\Sigma) - \mathcal{N}\bigl(Qy| x,\Sigma'
\bigr) \bigr|\leq C\bigl\|\theta-\theta'\bigr\|.
\]
Therefore, it is established that $ \|\Delta_{\theta,\theta'}^1
\|_\infty\leq
C\|\theta-\theta'\|^{1-2\kappa}$.

Let us consider the second term $\Delta^2_{\param,\param'}(x)$ in the
RHS of \eqref{e:sumDeltas}. Note first that if $x\in\Xset$ and
$y\in\cR_\theta(x)\cap\cR_{\theta'}(x)$, then by
\eqref{e:minMajGaussians}, $\pi(y)/\pi(x) \leq b/a$, so
\begin{eqnarray*}
\Delta^2_{\param,\param'}(x) &=& \int_{\cR_\theta(x)\cap\cR_{\theta'}(x)}
\frac{\pi(y)}{\pi(x)} \biggl\llvert\sum_{Q\in\mathcal{P}} \bigl(
\mathbh{1}_{V_\param}(y)\mathcal{N}(Qx| y,\Sigma) - \mathbh
{1}_{V_{\param'}}(y) \mathcal{N} \bigl(Qx| y,\Sigma' \bigr)
\bigr) \biggr\rrvert \,\mathrm{d}y
\\
&\leq& \frac{b}{a} \int_{\cR_\theta(x)\cap\cR_{\theta'}(x)} \biggl\llvert\sum
_{Q\in\mathcal{P}} \bigl( \mathbh{1}_{V_\param}(y)
\mathcal{N}(Qx| y,\Sigma) - \mathbh{1}_{V_{\param'}}(y)\mathcal {N} \bigl(Qx
| y,\Sigma' \bigr) \bigr) \biggr\rrvert \,\mathrm{d}y
.
\end{eqnarray*}
Therefore, repeating the above discussion for the bound of
$\Delta^1_{\theta,\theta'}(x)$, it is established that $ \|
\Delta_{\theta,\theta'}^2 \|_\infty\leq
C\|\theta-\theta'\|^{1-2\kappa}$.

To\vspace*{1.5pt} deal with $\Delta^3_{\param,\param'}(x)$, first observe that there exists
$C>0$ such that for any $\theta\in\cW_M$, $\theta'\in\cW_{M+1}$, and
$x,y\in\Xset$, we have
%
\begin{equation}\label{e:proposalRatios}
\biggl\llvert\frac{q_\theta(y,x)}{q_\theta(x,y)} - \frac{q_{\theta
'}(y,x)}{q_{\theta'}(x,y)} \biggr\rrvert\leq C\bigl\|
\theta-\theta'\bigr\|,
\end{equation}
because of \eqref{eq:def:qtheta}, \eqref{e:minMajGaussians} and the above
discussion for the upper bound of $\Delta_{\theta,\theta'}^1(x)$. Now let
$y\in\cA_\theta(x)\cap\cR_{\theta'}(x)$, then we have
\[
\frac{\pi(y)q_{\theta'}(y,x)}{\pi(x)q_{\theta'}(x,y)} \leq1 \leq\frac
{\pi(y)q_{\theta}(y,x)}{\pi(x)q_{\theta}(x,y)},
\]
which, combined with \eqref{e:proposalRatios}, yields
\[
1 - C\frac{\pi(y)}{\pi(x)}\bigl\|\theta-\theta'\bigr\|\leq
\frac{\pi(y)q_{\theta
'}(y,x)}{\pi(x)q_{\theta'}(x,y)} \leq1.
\]
Thus,
\begin{eqnarray*}
\Delta^3_{\theta,\theta'}(x) &=& \int_{\cA_\theta(x)\cap\cR_{\theta
'}(x)}
\biggl\llvert q_\theta(x, y)\mathbh{1}_{\Vt}(y) -
\frac{\pi(y)q_{\theta'}(y,x)}{\pi(x)q_{\theta'}(x,y)}q_{\theta
'}(x,y)\mathbh{1}_{V_{\theta'}}(y) \biggr
\rrvert \,\mathrm{d}y
\\
&\leq& \int \biggl( \bigl\llvert q_\theta(x, y)\mathbh{1}_{\Vt}(y)
- q_{\theta'}(x, y)\mathbh{1}_{V_{\theta'}}(y) \bigr\rrvert
\\
&&\hphantom{\int}{}\vee\cdots\vee \biggl| q_{\theta}(x, y)\mathbh{1}_{V_{\theta}}(y) -
q_{\theta'}(x, y)\mathbh{1}_{V_{\theta'}}(y)
\\
&&\hphantom{\int\vee\cdots\vee \biggl|}{}+ C\frac{\pi(y)}{\pi
(x)}\bigl\|
\theta-\theta'\bigr\| q_{\theta'}(x, y)\mathbh{1}_{V_{\theta'}}(y)
\biggr| \biggr) \,\mathrm{d}y . 
\end{eqnarray*}
Therefore, it is established that $ \|\Delta_{\theta,\theta'}^3
\|_\infty\leq C\|\theta-\theta'\|^{1-2\kappa}$.

The upper bound of $\Delta_{\theta,\theta'}^4(x)$ is similar, and thus
its proof
is omitted.
\end{pf}

\begin{lemm}[(Regularity in $\param$ of the solution of the Poisson
equation)]\label{l:poissonRegularity}
Let $M>0$ and $\kappa\in(0,1/2)$. Under Assumption~\ref{a:target}, there
exists $C>0$ such that for any $\theta\in\cW_M$ and $\theta'\in\cW_{M+1}$,
\[
\| P_{\param} \hat{H}_{\param} - P_{\param'}
\hat{H}_{\param'}\|_{\infty}\leq C\bigl\|\theta-
\theta'\bigr\|^{1-2\kappa}.
\]
\end{lemm}
\begin{pf}
We recall the following result, proved in \cite{FoMoPr12}, Lemma~5.5, page
24: there exists $C>0$ such that for any $\theta\in\cW_M$,
$\theta'\in\cW_{M+1}$, and $x\in\Xset$,
%
\begin{eqnarray}\label{e:decFromFoMoPr12}
&&\| P_{\param} \hat{H}_{\param} - P_{\param'}
\hat{H}_{\param'}\|_{\infty} \nonumber
\\
&&\quad\leq  C\bigl\| H(\cdot,\param)-H \bigl(
\cdot,\param' \bigr)\bigr\|_{\infty}
\\
&&\quad\quad{}+ C \sup
_{\param\in\cW_M }\bigl\| H(\cdot,\theta)\bigr\|_{\infty} \Bigl\{ \|
{\pi}_{\param}-{\pi}_{\param'}\|_{\TV}
  +\sup_{x\in\Xset} \bigl\| P_\theta(x,\cdot)
-P_{\theta'}(x,\cdot) \bigr\|_{\TV} \Bigr\} .\nonumber
\end{eqnarray}
Here, $\sup_{\param\in\cW_M }\|
H(\cdot,\theta)\|_{\infty}$ is finite by Lemma~\ref{l:lyapunov}.
Now, by Lemma~\ref{l:lyapunov} again, there
exists $C>0$ such that for any $\theta\in\cW_M$ and
$\theta'\in\cW_{M+1}$,
\[
\bigl\| H(\cdot,\param)-H \bigl(\cdot,\param' \bigr)
\bigr\|_\infty \leq  C \bigl\|\param-\param'\bigr\|.
\]
The upper bounds for the two last terms in the RHS of \eqref
{e:decFromFoMoPr12} result from Lemmas \ref{l:Regularity:theta:pi} and
\ref{l:Regularity:theta:kernel}, respectively.
\end{pf}

%


\subsection{Proof of Theorem \texorpdfstring{\protect\ref{th:thetaConvergence}}{3.2}}


We start by proving two lemmas.


\begin{lemm}
\label{l:controle:fluctuations}
Let $(\gamma_t)_{t>0}$ be a sequence such that $\sum_t \gamma_t^2 <
\infty$,
$\sum_t |\gamma_{t+1} - \gamma_t |< \infty$, and $\sum_t
\gamma_t^{2(1-\kappa)} < \infty$ for some $\kappa\in(0,1/2)$. Denote by
$\psi_t$ the value of the projection counter at the end of iteration
$t$, in
Algorithm \ref{figPseudocodeStableAMOR}. Let $(\theta_t, X_t)_{t\geq
0}$ be the
sequence generated by Algorithm \ref{figPseudocodeStableAMOR}. Under
Assumptions \ref{a:target} and \ref{a:limitSet}, for any $M>0$,
%
\begin{equation}
\label{eq:BIS} \lim_{L \to+\infty} \sup_{ \ell\geq1}
\Biggl\llVert \Biggl( \prod_{k=L}^{L+\ell}
\mathbh{1}_{\param_k \in\cW_M} \mathbh{1}_{\psi_{k+1} = \psi_k} \Biggr)\sum
_{k=L}^{L +\ell}\gamma_{k+1}
\bigl(H(X_{k+1}, \param_{k}) - h(\param_k)
\bigr) \Biggr\rrVert= 0\quad\quad \mbox{w.p.1} ,
\end{equation}
where $H$, $h$, $\tw$ and $\cW_M$ are given by \eqref{eq:def:H},
\eqref{def:h}, \eqref{e:defLyapunov} and \textup{\eqref{eq:def:levelset:lyap}},
respectively.
\end{lemm}
\begin{pf}
The proof is adapted from Theorem~2.7 in \cite{FoMoPr12}, and it is
thus omitted. It can be found in the supplemental article \cite{6}.
\end{pf}

\begin{lemm}
\label{l:stabilityWhenRecurrence}
Let $M \in(0,M_\star)$ and set
\[
\Gamma_{M_\star}^M = \bigl\{\theta\in\TsetR\dvtx  M_\star
\leq w(\theta)\leq M\bigr\} ,\quad\quad \iota= \inf_{\param\in\Gamma_{M_\star}^M} \bigl| \bigl
\langle\nabla w(\theta),h(\theta)\bigr\rangle\bigr|.\vadjust{\goodbreak}
\]
Under Assumptions
\ref{a:target} and \ref{a:limitSet}, there exist $\delta\in( 0, \iota)$ and
$\lambda,\beta>0 $ such that
\begin{enumerate}[(B)]
\item[(A)]   $ u\in\cW_{M_\star
}, 0\leq\gamma\leq\lambda,\|\xi\|\leq\beta
\Rightarrow w(u+\gamma\th(u)+\gamma\xi) \leq M$, and
\item[(B)]   $ u\in\Gamma
_{M_\star}^M,
0\leq\gamma\leq\lambda,\|\xi\|\leq\beta\Rightarrow w(u+\gamma
\th(u)+\gamma\xi) < \tw(u) - \gamma\delta$.
\end{enumerate}
\end{lemm}

\begin{pf} The proof is adapted from Lemma~2.1 in \cite{AnMoPr05},
and it is thus omitted. It can be found in the supplemental article \cite{6}.
\end{pf}

{\em Proof of item (1) in Theorem~\ref
{th:thetaConvergence}.} Let
$M>M_\star$, let $q$ (depending on $M$) be such that (see Remark~\ref{r:VtoK})
%
\begin{equation}
\label{eq:p:stability:tool1} \cW_M \subset\cW_{M+2} \subseteq
\TsetC_{\delta_q} ,
\end{equation}
and let $\param_0 \in\cW_M$. Let $\lambda,\beta$ be given by
Lemma~\ref{l:stabilityWhenRecurrence}. By Lemma~\ref{l:lyapunov}, $w$
and $h$
are uniformly continuous on $\cW_{M+1}$, and there exists $\eta>0$ such that
%
\begin{equation}
\label{e:etaprime} x \in\cW_{M}, \quad\quad\| x-y\|<\eta\Rightarrow
\bigl| w(x)- w(y)\bigr|< 1 \quad\mbox{and}\quad \bigl\| h(x)-h(y)\bigr\|< \beta.
\end{equation}
By Lemma~\ref{l:controle:fluctuations}, there exists an almost surely
finite r.v. $N$ such that w.p.1.,
%
\begin{equation}\label{e:Ncond1}
  n\geq N \Rightarrow\gamma_n \Bigl(1 + \sup_{x \in\Xset, \param\in
\cW_M}
\bigl\|H(x,\param)\bigr\| \Bigr) < \lambda\wedge\eta
\end{equation}
 and
\begin{equation}\label{e:eta}
\sup_{\ell\geq1} \Biggl( \prod_{i=N}^{N+\ell}
\mathbh{1}_{\param_i\in
\cW_{M+1} } \mathbh{1}_{\psi_{i+1} = \psi_i} \Biggr) \Biggl\llVert\sum
_{i =N}^{N +\ell} \gamma_{i+1} \bigl(
H(X_{i+1},\param_i) - h(\param_i) \bigr)
\Biggr\rrVert< \eta.
\end{equation}

The proof is by contradiction. Denote by $\psi_t $ the number of
projections at
the end of iteration~$t$. We assume that $\PP(\lim_t \psi_t = + \infty)>0$.
We can assume without loss of generality that
\[
w(\theta_N)\leq M ,\quad\quad \psi_N \geq q
\]
on the set $\{\lim_t \psi_t = +\infty\}$. Define the sequence
$(\theta'_{N+k})_{k\geq0}$ as
\[
\theta'_N=\theta_N \quad\mbox{and}\quad
\theta'_{N+k+1} = \theta'_{N+k} +
\gamma_{N+k+1}\th(\theta_{N+k}) .
\]
We prove by induction on $k$ that for
any $k\geq0$, on the set $\{\lim_t \psi_t = +\infty\}$,
\[
\theta_{N+k}'\in\cW_{M},\quad\quad
\theta_{N+k}\in\cW_{M+1} ,\quad\quad \bigl\|\theta'_{N+k}
- \theta_{N+k}\bigr\|< \eta,\quad\quad \psi_{N+k+1} = \psi_{N+k}
.
\]
The case $k=0$ is trivial since
$\theta'_N=\theta_N\in\cW_M$ and by using \eqref{e:etaprime}, \eqref
{e:Ncond1} and
\eqref{eq:p:stability:tool1} on the set $\{\lim_t \psi_t = + \infty\}
$. Assume
this property holds for $k\in\{0,1,\ldots,\ell\}$. Then we have
\[
\theta'_{N+\ell+1} = \theta'_{N+\ell} +
\gamma_{N+\ell+1}\th\bigl(\theta'_{N+\ell}\bigr) +
\gamma_{N+\ell+1} \bigl(\th(\theta_{N+\ell}) - \th\bigl(
\theta'_{N+\ell}\bigr) \bigr).
\]
Since
$\|\theta'_{N+\ell} - \theta_{N+\ell}\|< \eta$ and $ \theta
'_{N+\ell}$
is in $ \cW_{M}$, we have $\|\th(\theta'_{N+\ell}) -
\th(\theta_{N+\ell})\|< \beta$. Since $ \gamma_{N+\ell+1} < \lambda
$ by
\eqref{e:Ncond1}, we can apply Lemma~\ref{l:stabilityWhenRecurrence} to obtain
$\theta'_{N+\ell+1}\in\cW_{M}$. In addition,
\begin{eqnarray*}
\theta'_{N+\ell+1} - \theta_{N+\ell+1} & =& \sum
_{i=N}^{N+\ell} \gamma_{i+1}
\bigl(H(X_{i+1}, \param_i) - h(\param_i)
\bigr) \mathbh{1}_{\psi_{i+1} = \psi_i}
\\
&&{}+ \sum_{i=N}^{N+\ell}
\bigl(\gamma_{i+1} h(\param_i) + \param_i -
\param_0 \bigr) \mathbh{1}_{\psi_{i+1} \neq\psi_i}
\\
& = &\Biggl( \prod_{i=N}^{N +
\ell}
\mathbh{1}_{\param_i \in\cW_{M+1}} \Biggr)\sum_{i=N}^{N+\ell}
\gamma_{i+1} \bigl(H(X_{i+1}, \param_i) - h(
\param_i) \bigr) \mathbh{1}_{\psi_{i+1} = \psi_i} ,
\end{eqnarray*}
where we used the induction assumption in the last equality. From
\eqref{e:etaprime} and \eqref{e:eta}, this yields $\|\theta'_{N+\ell
+1} -
\theta_{N+\ell+1}\|<\eta$ and $\tw(\theta_{N+\ell+1})\leq M+1$.
Finally by
\eqref{e:etaprime}, equations \eqref{e:Ncond1} and \eqref{eq:p:stability:tool1}
imply that on the set $\{\lim_t \psi_t = + \infty\}$
\[
\param_{N+\ell} + \gamma_{N+\ell+1} H(X_{N+\ell+1},
\param_{N+\ell}) \in\cW_{M+2} \subset\TsetC_{\psi_{N+\ell}} ,
\]
that is, $\psi_{N+\ell+1} = \psi_{N+\ell}$. This concludes the induction.

As a consequence of this induction, we have $\psi_{N+\ell} = \psi_N$
for any $\ell
\geq0$ on the set $\{\lim_t \psi_t = + \infty\}$ which is a contradiction.

{\em Proof of item} (2) {\em in Theorem}~\ref{th:thetaConvergence}.
The proof is along the same lines as the proof of Theorem~2.3 of \cite{AnMoPr05}, page
5, 
and is thus omitted.

\subsection{Proof of Theorem \texorpdfstring{\protect\ref{th:xConvergence}}{3.3}}
\label{sec:thcvg}
The proof consists in checking the conditions of \cite{FoMoPr12}, Corollary~2.8. Let $f$ be a measurable bounded function.

By Lemma~\ref{l:poisson:bis}, (i) there exists a measurable function
$\hat{f}_\param$ such that $\hat{f}_\param- P_\param\hat{f}_\param=
f -
\pi_\param f$; and (ii) for any compact set $\cW_M$, there
exists $L$
(depending upon $M$) such that
\[
\forall\param\in\cW_M, x \in\Xset,\quad\quad \bigl|\hat{f}_\param(x)\bigr|
\leq L .
\]
By Theorem~\ref{th:thetaConvergence}, $\PP(\Omega_M) \uparrow1$ when
$M$ tends
to infinity where
%
\begin{equation}
\label{eq:def:OmegaM} \Omega_M = \bigcap_{t \geq0}
\{\param_t \in\cW_M \} .
\end{equation}
Therefore, in order to apply \cite{FoMoPr12}, Corollary~2.8, we only
have to
prove that almost surely,
%
\begin{eqnarray}
 \sum_k k^{-1} \sup
_{x \in\Xset} \bigl\|P_{\param_k}(x,\cdot) -P_{\param_{k-1}}(x,\cdot)
\bigr\|_{\TV} \mathbh{1}_{\Omega_M} &<& \infty, \label{eq:th:xcvg:tool1}
\\
 \lim_t \pi_{\param_t}(f) \mathbh{1}_{\Omega_M}
&= &\pi_{\param^\star}(f) \mathbh{1}_{\Omega_M}. \label{eq:th:xcvg:tool2}
\end{eqnarray}
By Lemma~\ref{l:Regularity:theta:kernel}, there exists $C$ and $\kappa
\in
(0,1/2)$ such that
\[
\sup_{x \in\Xset} \bigl\|P_{\param_k}(x,\cdot) -P_{\param_{k-1}}(x,
\cdot) \bigr\|_{\TV} \mathbh{1}_{\Omega_M} \leq C \|
\param_k-\param_{k-1}\|^{1-2\kappa} .
\]
In addition, by Theorem~\ref{th:thetaConvergence}, there exists a random
variable $K$, almost surely finite, such that for any $k \geq K$,
\[
\|\param_k-\param_{k-1}\| \mathbh{1}_{\Omega_M} \leq
\gamma_k \sup_{\param\in\cW_M,
x \in\Xset} \bigl|H(x,\param)\bigr| .
\]
This yields
\[
\sum_{k \geq K} k^{-1} \sup
_{x \in\Xset} \bigl\|P_{\param_k}(x,\cdot) -P_{\param_{k-1}}(x,\cdot)
\bigr\|_{\TV} \mathbh{1}_{\Omega_M} \leq C \sum
_{k \geq K} k^{-1} \gamma_k^{1-2\kappa} ,
\]
for some constant $C>0$. This concludes the proof of \eqref{eq:th:xcvg:tool1}.
The limit \eqref{eq:th:xcvg:tool2} is a consequence of
Lemma~\ref{l:Regularity:theta:pi}.

\begin{remark*} Note that in the proof above we use that the number of
random truncations is finite almost surely (when claiming that $\lim_M
\PP(\Omega_M) \uparrow1$) but only use the convergence of the sequence
$(\theta_t)_{t \geq0}$ in order to establish (\ref{eq:th:xcvg:tool2}).
When $f$ is such that $\pi_\param(f) = \pi(f)$ for any $\param\in
\TsetR$ (for example when $f$ is symmetric with respect to
permutations), then (\ref{eq:th:xcvg:tool2}) holds even if $(\param
_t)_{t \geq0}$ does not converge.
\end{remark*}

\subsection{Proof of Theorem \texorpdfstring{\protect\ref{th:ergodicity}}{3.4}}
\label{sec:thergo}
Let $f$ be a measurable function such that $\| f\|_{\infty}\leq
1$ and
set
\[
I_t(f) = \bigl| \mathbb{E}\bigl[f(X_t)
\mathbh{1}_{B}\bigr]-\pi_{\theta^\star}(f)\mathbb{P}(B) \bigr| = \bigl|
\mathbb{E}\bigl[ \bigl( f(X_t) - \pi_{\theta^\star}(f) \bigr)
\mathbh{1}_{B}\bigr] \bigr|,
\]
where $ B = \{ \lim_q \param_q = \param_\star\}$. Let $\eps>0$. We prove
that there exists $T_{\eps}$ such that for all $t\geq T_{\eps}$, $\sup_{\{f:
\|f\|_\infty\leq1 \}} I_t(f)\leq4\eps$. Choose $\kappa\in(0,1/2)$ and
$\delta>0$ such that
%
\begin{equation}
C_{M_\star+1}\delta^{1-2\kappa}\leq\eps, \label{e:ergoTool1}
\end{equation}
where $M_\star$ and $C_{M_\star}$ are defined in Assumption~\ref{a:limitSet}
and in Lemma~\ref{l:Regularity:theta:pi}, respectively. Choose $r_\eps$
such that
%
\begin{equation}
2(1-\rho_{M_\star+1})^{r_\eps}\leq\eps, \label{e:ergoTool2}
\end{equation}
where $\rho_{M_\star+1}$ is defined in Lemma~\ref{l:poisson:bis}. By uniform
continuity of $w$ on $\cW_{M_\star+2}$, assume finally $\delta$ is small
enough that
%
\begin{equation}
\theta\in\cW_{M_\star+1},\theta'\in\Theta,\quad\quad \bigl\|\theta-
\theta'\bigr\|\leq\delta\Rightarrow\bigl| w(\theta)-w \bigl(
\theta' \bigr)\bigr|\leq\frac{1}{r_\eps+1}. \label{e:ergoTool3}
\end{equation}
There exists $T_\varepsilon^1$ such that for any $t \geq T_\varepsilon^1$,
\[
\PP \Bigl( \bigl\|\param_{t-r_\varepsilon} - \param^\star\bigr\| \leq\delta, \lim
_q \param_q = \param^\star \Bigr) \leq
\varepsilon/2 .
\]
Hence, for any $t \geq T_\varepsilon^1$, $I_t(f)\leq\sum_{i=1}^3
I_t^i(f) +
\varepsilon$, where
%
\begin{eqnarray}
I_t^1(f) &=& \bigl|\mathbb{E} \bigl[ \bigl(
f(X_t) - P_{\theta_{t - r_\eps}}^{r_\eps} f(X_{t-r_\eps})
\bigr)\mathbh{1}_{\|\param_{t-r_\varepsilon} - \param^\star\| \leq
\delta} \bigr] \bigr|,\label{e:defIt1}
\\
I_t^2(f) &=& \bigl|\mathbb{E} \bigl[ \bigl(
P_{\theta_{t-r_\eps}}^{r_\eps} f(X_{t-r_\eps}) - \pi_{\theta_{t-r_{\eps}}}(f)
\bigr)\mathbh{1}_{\|\param_{t-r_\varepsilon} - \param^\star\| \leq
\delta} \bigr] \bigr|,\label{e:defIt2}
\\
I_t^3(f) &=& \bigl|\mathbb{E} \bigl[ \bigl(
\pi_{\theta_{t-r_{\eps}}}(f) - \pi_{\theta^\star}(f) \bigr)\mathbh {1}_{\|\param_{t-r_\varepsilon} - \param^\star\| \leq\delta}
\bigr] \bigr|. \label{e:defIt3}
\end{eqnarray}
We first upper bound $I_t^1(f)$. For $\theta,\theta'\in\Theta$, let
\[
D\bigl(\theta,\theta'\bigr) = \sup_{x\in\Xset} \bigl\|
P_\theta(x,\cdot) - P_{\theta'}(x,\cdot) \bigr\|_{\TV}.
\]
Applying \cite{AtFoMoPr11}, Proposition~1.3.1, it comes for any $t \geq
T_\varepsilon^1$,
\begin{eqnarray*}
I_t^1 &\leq& \mathbb{E} \Biggl[ 2\wedge\sum
_{j=1}^{r_\eps-1} D(\theta_{t-r_\eps+j},
\theta_{t-r_\eps}) \mathbh{1}_{\|\param_{t-r_\varepsilon} - \param
^\star\| \leq\delta} \Biggr]
\\
&\leq& \mathbb{E} \Biggl[ 2\wedge\sum_{j=1}^{r_\eps-1}
(r_\eps-j) D(\theta_{t-r_\eps+j}, \theta_{t-r_\eps+j-1})
\mathbh{1}_{\|\param_{t-r_\varepsilon} - \param^\star\| \leq\delta} \Biggr] ,
\end{eqnarray*}
where we used that for any $q,\ell>0$ $ D(\theta_{q+\ell},\theta_q) \leq
\sum_{j=1}^\ell D(\theta_{q+j},\theta_{q+j-1})$.
By Theorem \ref{th:thetaConvergence}, the random iteration number $\tau
_{\psi}$ where the last
projection occurs in Algorithm \ref{figPseudocodeStableAMOR} is
finite with probability one. Let then $M_\eps$ be such that
$2\mathbb{P}(\tau_\psi\geq M_{\eps})\leq\eps/2$, so that
\[
I_t^1(f) \leq\mathbb{E} \Biggl[ 2\wedge\sum
_{j=1}^{r_\eps-1} (r_\eps-j) D(
\theta_{t-r_\eps+j}, \theta_{t-r_\eps+j-1})\mathbh{1}_{\|\param
_{t-r_\varepsilon} -
\param^\star\| \leq\delta}
\mathbh{1}_{\tau_\psi\leq M_\eps} \Biggr] + \frac{\eps}{2}.
\]
Let now $T_\eps^2\geq T_\eps^1 \vee(M_\eps+r_\eps)$ be such that
\[
t\geq T_\eps^2 \Rightarrow\gamma_t \sup
_{x\in\Xset,
\theta\in\cW_{M_\star+2}}\bigl\| H(x,\theta)\bigr\|\leq \delta.
\]
Then, by recurrence and using \eqref{e:ergoTool3}, we obtain that on
$\{\|\theta_{t-r_\eps}-\theta_\star\|\leq\delta\}$,
$\theta_{t-r_\eps+j}\in\cW_{M_\star+1}$ for all $0\leq j\leq r_\eps$. By
Lemma~\ref{l:Regularity:theta:kernel}, this yields for any $t \geq
T_\varepsilon^2$
\[
I_t^1(f) \leq C_{M_\star+1} \Bigl[\sup
_{x\in\Xset, \theta\in\cW_{M_\star
+2}}\bigl\| H(x,\theta)\bigr\|\Bigr]^{1-2\kappa} \sum
_{j=1}^{r_\eps-1} (r_\eps-j)
\gamma_{t-r_\eps+j}^{1-2\kappa} + \frac{\eps}{2},
\]
and there exists $T_\eps^3\geq T_\eps^2$ such that $t\geq
T_\eps^3\Rightarrow\sup_{\{f: \|f\|_\infty\leq1 \}} I_t^1(f)\leq\eps$.

We now consider $I_t^2(f)$; it holds
\[
I_t^2 \leq \mathbb{E} \bigl[ \bigl\llVert
P_{\theta_t-r_\eps}^{r_\varepsilon} (X_{t-r_\eps},\cdot) - \pi_{\theta
_{t-r_\eps}}
\bigr\rrVert_{\TV} \mathbh{1}_{\|\theta_{t-r_\eps}-\theta^\star\|\leq\delta} \bigr].
\]
By \eqref{e:ergoTool3}, $\|
\theta_{t-r_\eps}-\theta^\star\|\leq\delta\Rightarrow
\theta_{t-r_\eps}\in\cW_{M_\star+1}$ and thus, applying Lemma~\ref
{l:poisson:bis} and \eqref{e:ergoTool2}
\[
\sup_{\{f: \|f\|_\infty\leq1 \}} I_t^2(f) \leq 2(1-
\rho_{M_\star+1})^{r_\eps} \leq\varepsilon.
\]
The derivation of the upper bound of $I_t^3$ is similar to that of $I_t^2$,
with Lemma~\ref{l:poisson:bis} replaced by Lemma~\ref{l:Regularity:theta:pi}
and uses (\ref{e:ergoTool1}). Details are omitted.

\begin{remark*}\label{r2}
The proof above can be easily adapted (details are omitted) to address the
case when (i) $(\theta_t)_{t \geq0}$ is stable but does not
necessarily converges, and (ii) the function $f$ is bounded and
satisfies $\pi_\param(f) = \pi(f)$ for any $\param\in\TsetR$. The main
ingredients for this extension are to replace $\mathbh{1}_B$ with the constant
function $\mathbh{1}$, and to replace the set $\{\|
\theta_{t-r_\eps}-\theta^\star\|\leq\delta\}$ with $\{\theta_{t -
r_\eps} \in\cW_{M_\star} \}$. Since the sequence is stable, $\lim_M
\PP(\Omega_M) \uparrow1$ where $\Omega_M$ is given by (\ref{eq:def:OmegaM}).
$M_\star$ is chosen so that $\mathbb{E} [ \llvert f(X_t) - \pi
(f)\rrvert
\mathbh{1}_{\Omega_{M_\star}} ] \leq\varepsilon$. We then obtain, for
such a
function $f$,
\[
\lim_{t \to\infty} \mathbb{E} \bigl[f(X_t) \bigr] =
\pi(f) .
\]
\end{remark*}

\end{appendix}

\section*{Acknowledgements}
This work was supported by the ANR-2010-COSI-002 grant of the French National
Research Agency.
\begin{supplement}
\stitle{Long version of the paper}
\slink[doi]{10.3150/13-BEJ578SUPP} 
\sdatatype{.pdf}
\sfilename{BEJ578\_supp.pdf}
\sdescription{This long version of the paper features an additional
evaluated method for Section~\ref{s:example} (AM with posterior
reordering), examples of the behavior of AMOR on a nonlinear
symmetrized unimodal distribution and on a genuinely bimodal distribution,
and complete proofs.}
\end{supplement}

%

\printhistory

\end{document}